%% file: main.tex
\documentclass[a4paper,UKenglish,cleveref, autoref, final, thm-restate]{lipics-v2021}
\input{macros}

\usepackage{tikz-cd}

\title{Bounding the Escape Time of a Linear Dynamical System over a Compact Semialgebraic Set} %TODO Please add

\titlerunning{Bounding the Escape Time of an LDS over a Compact Semialgebraic Set} %TODO
 \author{Julian D'Costa}{Department of Computer Science, University of Oxford, UK}%
 {julianrdcosta@gmail.com}{}{emmy.network foundation under the aegis of the Fondation de Luxembourg.}

 \author{Engel Lefaucheux}{Max Planck Institute for Software Systems, Saarland Informatics Campus, Germany\\
Université de Lorraine, Inria, LORIA, Nancy, France
 }{elefauch@mpi-sws.org}{https://orcid.org/0000-0003-0875-300X}{}
 \author{Eike Neumann}{
   Swansea University, Swansea, UK 
   }{neumaef1@gmail.com}{}{}

 \author{Jo\"el Ouaknine}%
 {Max Planck Institute for Software Systems, Saarland Informatics Campus, Germany}%
 {joel@mpi-sws.org}{https://orcid.org/0000-0003-0031-9356}{%
 ERC grant AVS-ISS
 (648701), and DFG grant 389792660 as part of TRR 248 (see
 \url{https://perspicuous-computing.science}).\\
 Jo\"el Ouaknine is also affiliated with Keble College, Oxford as \href{http://emmy.network/}{\texttt{emmy.network}} Fellow.}

 \author{James Worrell}%
 {Department of Computer Science, University of Oxford, UK}%
 {jbw@cs.ox.ac.uk}{}{EPSRC Fellowship EP/N008197/1.}

\authorrunning{J. D'Costa, E. Lefaucheux, E. Neumann, J. Ouaknine, and J. Worrell} 
\Copyright{Julian D'Costa, Engel Lefaucheux, Eike Neumann, Jo\"el Ouaknine, James Worrell}
 %TODO mandatory: Please choose ACM 2012 classifications from https://dl.acm.org/ccs/ccs_flat.cfm 

 %TODO mandatory; please add comma-separated list of keywords

\category{} %optional, e.g. invited paper

\supplement{}%optional, e.g. related research data, source code, ... hosted on a repository like zenodo, figshare, GitHub, ...

\nolinenumbers %uncomment to disable line numbering

\hideLIPIcs  %uncomment to remove references to LIPIcs series (logo, DOI, ...), e.g. when preparing a pre-final version to be uploaded to arXiv or another public repository

\usepackage{float}

\begin{document}
	
	\maketitle

	\begin{abstract}
		We study the Escape Problem for discrete-time linear dynamical systems over compact semialgebraic sets.
		We establish a uniform upper bound on the number of iterations it takes for every orbit of a rational matrix to escape a compact semialgebraic set
		defined over rational data.
		Our bound is doubly exponential in the ambient dimension, singly exponential in the degrees of the polynomials used to define the semialgebraic set,
		and singly exponential in the bitsize of the coefficients of these polynomials and the bitsize of the matrix entries.
		We show that our bound is tight by providing a matching lower bound.
	\end{abstract}

	%\newpage

\input{introduction}
\input{preliminaries}
	\input{bound-overview}	
	\input{kronecker-en}
	\input{recurrent-eigenspace}
	\input{nonrecurrent-eigenspace}

	\input{combined-bound}

	\input{double-exp-example}

	%%
	%% Bibliography
	%%
	
	%% Please use bibtex, 
	\bibliography{refs}
	\input{app-real-jordan.tex}	
	\input{app-JNF}
	\input{app-Mignotte}
	\input{app-Kronecker}
	\input{app-Separation}
	\input{app-Recurrent}
	\input{app-Nonrecurrent}
	\input{app-example}
	
\end{document}

%% file: macros.tex
\usepackage{thmtools} % used for restating theorems
\usepackage{thm-restate}  % used for restating theorems
\usepackage[colorinlistoftodos,prependcaption,textsize=tiny,obeyDraft]{todonotes}
\usepackage{microtype}%if unwanted, comment out or use option "draft"

\usepackage{amsmath,amssymb}
\usepackage{mathtools}
\usepackage{xcolor,pgf,tikz}
\usepackage{comment}
\usepackage{array}
\usepackage{makecell}
\usepackage[ruled,vlined]{algorithm2e}
\usetikzlibrary{arrows,automata,shapes,decorations}

\newtheorem*{lemma*}{Lemma}
\newtheorem*{problem*}{Problem}
\newcommand{\norm}[1]{\left\lVert#1\right\rVert}
\newcommand{\Jnorm}[1]{\norm{#1}_{J}}
\newcommand{\JBall}{B_{J}}

\newcommand{\N}{\mathbb{N}}

\newcommand*{\hull}[1]{\mathrm{Co}\left(\kern0pt#1 \right)}
\newcommand*{\cone}[2][{}]{\mathrm{Cone}_{#1}\left(\kern0pt#2 \right)}
\newcommand*{\interior}[1]{ {\kern0pt#1}^{\mathrm{o}}  }
\newcommand*{\closure}[1] { \overline { {\kern0pt#1} } }

\newcommand{\modul}{\gamma}

\newcommand{\R}{\mathbb{R}}

\def\eps{\varepsilon}

\def\Z{\mathbb{Z}}

\def\C{\mathbb{C}}

\renewcommand{\O}{\mathcal{O}}
\newcommand{\Set}[2]{\left\{ #1 \; \mid \; #2 \right\}}
\newcommand{\clos}[1]{\overline{#1}}
\newcommand{\Q}{\mathbb{Q}}

%% file: introduction.tex
\section{Introduction}

An \emph{invariant set} of a dynamical system is a set $K$ such that every
trajectory that starts in $K$ remains in $K$.  Dually, an \emph{escape set}
$K$ is one such that every trajectory that starts in $K$ eventually leaves
$K$ (either temporarily or permanently).  While it is usually
straightforward to establish that a given set $K$ is invariant, it can
be challenging to decide whether it is an escape set.  Indeed, while the
former problem amounts to showing that $K$ is closed under the transition
function, the latter potentially involves considering entire orbits.  In
particular, even in case $K$ has a finite escape time (the maximum
number of steps for an orbit to escape the set), it can be highly
non-trivial to establish an explicit upper bound on the escape
time.

In this paper we focus on escape sets for (discrete-time) linear dynamical systems.
Given a rational matrix $A\in \Q^{n \times n}$ we say that
$K \subseteq \R^n$ is an escape set for $A$ if for all points
$x\in K$, there exists $t\in\N$ such that $A^tx\not\in K$.  The
\emph{compact escape problem (CEP)} asks to decide whether a given
compact semialgebraic set $K$ is an escape set for a given matrix $A$.
Decidability of CEP was shown in~\cite{NOW20} and its computational
complexity was characterised in~\cite{DLNOW21} as being interreducible
with the decision problem for a certain fragment of the theory of real
closed fields.

The present paper focusses exclusively on positive instances $(A,K)$
of CEP, that is, we assume that we are given a compact semialgebraic
escape set for a linear dynamical system.  In this situation it turns
out, due to compactness of $K$, that there exists a finite time $T$
such that for all $x \in K$ there exists $t \leq T$ with
$A^t x \not \in K$.  The least such $T$ is called the \emph{escape
  time} of $(A,K)$.  Our main result (Theorem~\ref{Theorem: main
  theorem}, shown below) gives an explicit upper bound on the escape
time of $(A,K)$ as a function of the length of the description of the
matrix $A$ and semialgebraic set $K$.  In general, it is recognised
that bounded liveness is a more useful property than mere liveness.
Theorem~\ref{Theorem: main theorem} can be used to establish bounded
liveness of several kinds of systems.  For example, the result gives
an upper bound on the termination time of a single-path linear loop
with compact guard (cf.~\cite{Tiw04,Braverman2006}); it also gives a
bound on the number of steps to remain in a particular control
location of a hybrid system before a given (compact) state invariant
becomes false, forcing a transition.

We next introduce some terminology to formalise our main contribution.
We say that a semialgebraic set $S$ has complexity at most
$(n, d, \tau)$ if it can be expressed by a boolean combination of
polynomial equations and inequalities $P(x_1,\dots,x_n) \bowtie 0$
with $\bowtie \in \{\leq, =\}$,
involving polynomials $P \in \Z[x_1,\dots,x_n]$ in at most $n$
variables of total degree at most $d$ with integer coefficients
bounded in bitsize by $\tau$.  Our main result is as follows:
\begin{theorem}\label{Theorem: main theorem}
  There exists an integer function
  $ \operatorname{CompactEscape}(n,d,\tau) \in
  2^{\left(d\tau\right)^{n^{O(1)}}} $ with the following property.  If
  $K \subseteq \R^n$ is a compact semialgebraic set of complexity at
  most $(n,d,\tau)$ that is an escape set for a matrix
  $A \in \Q^{n \times n}$ with entries of bitsize at most $\tau$, then
  the escape time of $K$ is bounded by
  $\operatorname{CompactEscape}(n, d, \tau)$.
\end{theorem}

As explained in the proof sketch below, Theorem \ref{Theorem: main
  theorem} relies on the availability of certain quantitative bounds
within semialgebraic geometry and number theory, particularly
concerning quantifier elimination and Diophantine approximation.
The latter results are crucial to handling the case in which the
matrix $A$ has complex eigenvalues of absolute value one.

Note that the upper bound on the escape time in Theorem~\ref{Theorem:
  main theorem} is singly exponential in the degrees and the bitsize
of the coefficients of the polynomials used to define $K$ and the
bitsize of the coefficients of $A$.  It is doubly exponential in the
dimension.  In Section~\ref{Section: Lower Bounds} we provide two
examples, one where $A$ is an isometry and another in which all
eigenvalues of $A$ have absolute value strictly greater than one, that
yield a corresponding lower bound of this form.  It is moreover
straightforward to give examples of non-compact escape sets for which
the escape time is infinite.

%\EN{TODO: maybe a comment on 
%  linear vs. non-linear polynomial in $n$ in second exponent} \EL{I 
%  wouldn't do that, but we can see with Joël.}
%A positive answer to the CEP guarantees that the system cannot remain 
%stuck in $K$.  
%However, it is often not enough to know the system will 
%eventually exit the set, we also need to know when. 
%Given a positive instance $(A,K)$ of the CEP, we are interested in
%computing a bound $N$ such that every point of $K$ exits the set
%before $N$ iterations of $A$. Note that the existence of such a bound
%is ensured by the fact that $K$ is compact.

%Main contributions.
%We show how to construct a quantitative Kronecker version that answers the questions "how much time does it take a trajectory whose closure is a semialgebraic set to get within epsilon of a specific point in that closure. (useful tool for quantitative analysis of dynamical systems)

%we use the above and dsicrete versions of lemmas from compact polytope paper to get bounds for escape for arbitrary semialgebraic sets

%(unlike compact polytope paper) We show the bound is tight by
%demonstrating a matching lower bound

\textbf{Proof Overview.}
Let us now give a high-level overview of the
proof of Theorem~\ref{Theorem: main theorem}.  As in the statement of
the theorem, let $K \subseteq \R^n$ be a compact semialgebraic set of
complexity at most $(n,d,\tau)$ and let $A \in \Q^{n \times n}$ be a
matrix with entries of bitsize bounded by $\tau$, and such that for
all $x \in K$ there exists $t \in \N$ such that $A^t x \notin K$.

To facilitate the analysis of the dynamical behaviour of $A$ we first
transform our system into real Jordan normal form.  A theorem of Cai
\cite{Cai94} ensures that this step does not significantly increase
the complexity of the system.

The dynamics of $A$ naturally decomposes into a rotational part,
corresponding to eigenvalues of modulus one, and an expansive or
contractive part, corresponding to eigenvalues of absolute value
different from $1$ and to generalised eigenvalues of arbitrary moduli.
Accordingly, the ambient space $\R^n$ decomposes into two subspaces
$V_{\operatorname{rec}}$ and $V_{\operatorname{non-rec}}$, such that
$A$ exhibits rotational behaviour on $V_{\operatorname{rec}}$ and
expansive or contractive behaviour on $V_{\operatorname{non-rec}}$.  
We start by considering the special cases where either
$V_{\operatorname{rec}} = 0$ or $V_{\operatorname{non-rec}} = 0$, so
that only one of the two types of behaviours occurs.

First, assume that $A$ has no complex eigenvalues of modulus $1$.
Since every trajectory under $A$ escapes $K$ we have in particular
that $0 \notin K$.  A theorem due to Jeronimo, Perrucci and
Tsigaridas~\cite{Jeronimo} shows that $K$ is bounded away from zero by
a function of the form $2^{-(d\tau)^{n^{O(1)}}}$ and a theorem due to
Vorobjov \cite{Vorobjov84} establishes an upper bound on the absolute
value of every coordinate of every point in $K$ of the form
$2^{\left(d\tau\right)^{n^{O(1)}}}$.  Furthermore, thanks to a result
of Mignotte \cite{Mig82}, we can bound the eigenvalues of $A$ away
from $1$ by a function of the form $2^{\tau^{n^{O(1)}}}$.  This yields
a doubly exponential bound on how long it takes for $A$ to leave the
set $K$ (either by converging to 0 or by converging to infinity in
some eigenspace).

Now assume that all eigenvalues of $A$ have modulus $1$.  This case is
handled through a combination of two bounds.  For the first bound we
start by noting that for every $x \in K$ the closure of the orbit
$\clos{\O_A(x)}$ is a compact semialgebraic set that is not entirely
contained within $K$.  In fact we show that for all $x \in K$ there
exists a point $y \in \clos{\O_A(x)}$ whose distance to $K$ is at
least $2^{-\left(d \tau\right)^{n^{O(1)}}}$.  This bound is achieved
by applying~\cite[Theorem 1]{Jeronimo} to a suitable polynomial on an
auxiliary semialgebraic set, which is constructed using quantifier
elimination.  The singly exponential bounds obtained in
\cite{HeintzRoySolerno90,Renegar1992} are crucial for this step to
work.  The second step of the argument combines Baker's theorem on
linear forms in logarithms with a quantitative version of Kronecker's
theorem on simultaneous Diophantine approximation to obtain a bound of
the form $N_P \in 2^{\left(\tau P\right)^{n^{O(1)}}}$ such that for
all positive integers $P$ every point $z \in \clos{\O_A(x)}$ is within
$2^{-P}$ of a point of the form $A^t x $ with $0 \leq t \leq N_P$.
Combining the two bounds described above, we obtain a doubly
exponential bound on the escape time.

In the presence of both types of behaviour, the analysis of each case
becomes more involved.  We select a parameter $\varepsilon > 0$ and
partition $K$ into three sets:
$K_{\operatorname{rec}} = K \cap V_{\operatorname{rec}}$,
$K_{\geq \varepsilon}$, and $K_{< \varepsilon}$.  The matrix $A$
exhibits purely rotational behaviour on $K_{\operatorname{rec}}$.
Intuitively, on $K_{\geq \varepsilon}$ the expansive or contractive
behaviour of $A$ dominates the overall dynamics, while on
$K_{< \varepsilon}$ the rotational behaviour dominates the overall
dynamics.  We establish in Lemma~\ref{Lemma: recurrent escape bound} a
bound $N_{\operatorname{rec}}$ such that for each initial point
$x \in V_{\operatorname{rec}}$, one of its first
$N_{\operatorname{rec}}$ iterates is bounded away from $K$.  In Lemma
\ref{Lemma: non-recurrent overall bound} we establish a bound
$N_{\geq \varepsilon}$ such that every $x \in K_{\geq \varepsilon}$
either escapes or enters
$K_{< \varepsilon} \cup K_{\operatorname{rec}}$ within at most
$N_{\geq \varepsilon}$ iterations.  Finally, in Section \ref{Section:
  proof of main theorem}, we establish a bound on how often the system
can switch from a state where rotational behaviour dominates to one
where expansive or non-expansive behaviour does and vice versa.  We
use this to combine the two bounds to an overall bound on the escape
time, proving Theorem \ref{Theorem: main theorem}.

\textbf{Main Contributions.}
While decidability of CEP was already established in \cite{NOW20}, the proof given there was non-effective, combining two unbounded searches.
To obtain a uniform quantitative bound on the escape time, the argument given in \cite{NOW20} needs to be refined and extended in two significant ways:

    Firstly, one needs to establish non-trivial quantitative refinements of the techniques used in the decidability proof:
    to bound the escape time for purely expanding or retracting systems, 
    we need to combine the sharp effective bounds on compact semialgebraic sets from real algebraic geometry established in \cite{Vorobjov84, Jeronimo}
    with Mignotte's root separation bound \cite{Mig82}.
    The case of purely rotational systems requires an original combination of a quantitative version of Kronecker's theorem on simultaneous Diophantine approximation \cite{QuantKroneckerSurvey} and a quantitative version of Baker's theorem on linear forms in logarithms \cite{BakerWustholz}.   
    All of these techniques were completely absent from the decidability proof.

    Secondly, to establish mere decidability of the problem, it was possible to study the possible behaviours of the system -- rotating, expanding, or retracting -- in isolation.
    For example, if the set $K$ contains a point which has a non-zero component in an eigenspace of $A$ for an eigenvalue whose modulus is strictly greater than one, 
    then the system must eventually escape.
    However, no uniform bound on the escape time may be derived in this situation, for the component is allowed to be arbitrarily close to zero.
    Therefore, as outlined above, it is necessary in our proof to subdivide $K$ into pieces where rotational, retractive, and expansive behaviour can be present simultaneously.
    The interaction of the three behaviours significantly increases the difficulty of the analysis and requires completely new ideas.

%% file: preliminaries.tex
\section{Mathematical Tools}

We use the
following singly exponential quantifier elimination result given in
\cite{BasuPollackRoy}.  For a historical overview on this type of result
see \cite[Chapter 14, Bibliographical Notes]{BasuPollackRoy}.

\begin{theorem}[{\cite[Theorem 14.16]{BasuPollackRoy}}]\label{Theorem: singly exponential quantifier elimination precise}
    Let 
    $S \subseteq \R^{k + n_1 + \dots + n_{\ell}}$ 
    be a semialgebraic set of complexity at most $(k + n_1 + n_2 + \dots + n_{\ell},d,\tau)$.
    Let $Q_1,\dots,Q_\ell \in \{\exists, \forall\}$ be a sequence of alternating quantifiers.
    Consider the set $S' \subseteq \R^k$ of all $(x_1,\dots,x_k) \in \R^{k}$ satisfying the first-order formula
    \begin{align*}
        \left(Q_1 (x_{1,1},\dots,x_{1,n_1})\right). &\dots \left(Q_{\ell} (x_{\ell,1},\dots,x_{\ell,n_{\ell}})\right).\\
        &\left(
                    \left(x_1,\dots,x_k, x_{1,1},\dots,x_{1,n_1},\dots,x_{\ell,1},\dots,x_{\ell, n_{\ell}}\right)
                    \in S    
        \right)
    \end{align*}
    Then $S'$ is a semialgebraic set of complexity at most $(k, d^{O(n_1\cdot\dots\cdot n_{\ell})}, \tau d^{O(n_1\cdot\dots\cdot n_{\ell}\cdot k)})$.
\end{theorem}

The next theorem is due to Vorobjov \cite{Vorobjov84}.
See also \cite[Lemma 9]{GriorievVorobjov88} and \cite[Theorem 4]{BasuRoyRadiusBound}.

\begin{theorem}\label{Theorem: Basu-Roy radius bound}
    There exists an integer function 
    $
        \operatorname{Bound}(n,d,\tau) \in 2^{\tau d^{O(n)}}
    $
    with the following property:

    Let $K$ be a compact semialgebraic set of complexity at most $(n, d, \tau)$.
    Then $K$ is contained in a ball centred at the origin of radius at most 
    $
        \operatorname{Bound}(n,d,\tau)
    $.
\end{theorem}

A closely related result, due to \cite{Jeronimo}, yields a lower bound on the minimum of a polynomial over a compact semialgebraic set, provided the minimum is non-zero. The result in \cite{Jeronimo} mentions explicit constants, which is more than we need.

\begin{theorem}[{\cite[Theorem 1]{Jeronimo}}]\label{Theorem: Jeronimo}
    There exists an integer function 
    $
        \operatorname{LowerBound}(n,d,\tau) \in 2^{\left(\tau d\right)^{n^{O(1)}}}
    $
    such that the following holds true: 

    Let $P \in \Q[x_1,\dots,x_n]$ be a polynomial of degree at most $d$, whose coefficients have bitsize at most $\tau$.
    Let $K$ be a compact semialgebraic set of complexity at most $(n,d,\tau)$.
    If $\min_{x \in K} P(x) > 0$ then 
    $\min_{x \in K} P(x) > 1/\operatorname{LowerBound}$.
\end{theorem}

With the help of Theorem \ref{Theorem: singly exponential quantifier elimination precise}, 
Theorem \ref{Theorem: Jeronimo} can be generalised to yield a lower bound on the distance of two disjoint compact semialgebraic sets.
A very similar result is proved in \cite{SS17} under more general assumptions. 
Unfortunately, the complexity bound stated there is not sufficiently fine-grained for our purpose, since the author do not distinguish the dimension of a set from the other complexity parameters.

\begin{lemma}\label{Lemma: compact semialgebraic set separation bound}
    There exists an integer function 
    $
        \operatorname{Sep}(n, d, \tau) \in 2^{\left(\tau d\right)^{n^{O(1)}}}
    $
    with the following property:

    Let $K$ and $L$ be compact semialgebraic sets of complexity at most $(n,d,\tau)$.
    Assume that every $x \in K$ has positive euclidean distance to $L$.
    Then $\inf_{x \in K} d(x,L) > 1/\operatorname{Sep}(n,d,\tau)$.
\end{lemma}
\begin{proof}
    See Appendix \ref{Appendix: separation}.
\end{proof}

We require a version of Kronecker's theorem on simultaneous Diophantine approximation.
See \cite[Corollary 3.1]{10.5555/2634074.2634101} for a proof.

\begin{theorem}
\label{Theorem: Kronecker}
    Let $(\lambda_1, \dots, \lambda_m)$ be complex algebraic numbers of modulus $1$.
    Consider the free Abelian group
    \[ 
        L = \Set{(n_1,\dots,n_m) \in \Z^m}{\lambda_1^{n_1}\cdot\dots\cdot \lambda_m^{n_m} = 1}.
    \]
    Let $(\beta_1,\dots,\beta_s)$ be a basis of $L$.
    Let $\mathbb{T}^m = \Set{(z_1,\dots,z_m) \in \C^m}{|z_j| = 1}$ denote the complex unit $m$-torus.
    Then the closure of the set 
    $\Set{(\lambda_1^k,\dots,\lambda_m^k) \in \mathbb{T}^m}{k \in \N}$
    is the set 
    $S = \Set{(z_1,\dots,z_m) \in \mathbb{T}^m}{\forall j\leq s. (z_1,\dots,z_m)^{\beta_j} = 1}$.

    Moreover, for all $\varepsilon > 0$ and all $(z_1,\dots,z_m) \in S$ 
    there exist infinitely many indexes $k$ such that 
    $|\lambda_j^k - z_j| < \varepsilon$
    for $j = 1,\dots,m$.
\end{theorem}

Moreover, the integer multiplicative relations between given complex algebraic numbers in the unit circle can be elicited in polynomial space.
For a proof see \cite{CaiLiptonZalcstein00, masser88}.
We assume the standard encoding of algebraic numbers, see \cite{Cohen} for details.

\begin{theorem}\label{Theorem: Masser}
    Let $(\lambda_1, \dots, \lambda_m)$ be complex algebraic numbers of modulus $1$.
    Consider the free Abelian group
    \[ 
        L = \Set{(n_1,\dots,n_m) \in \Z^m}{\lambda_1^{n_1}\cdot\dots\cdot \lambda_m^{n_m}}.
    \]
    Then one can compute in polynomial space a basis $(\beta_1,\dots,\beta_s) \in (\Z^m)^s$ for $L$.
    Moreover, the integer entries of the basis elements $\beta_j$ are bounded polynomially in the 
    size of the encodings of $\lambda_1,\dots,\lambda_m$ and singly exponentially in $m$.
\end{theorem}

We need to be able to bound away the modulus of eigenvalues that fall outside the unit circle 
away from $1$.
This is achieved by combining a classic result due to Mignotte \cite{Mig82} on the separation 
of algebraic numbers with a bound on the height of the resultant of two polynomials,
proved in \cite[Theorem 10]{ResultantBound}.

\begin{lemma}\label{Lemma: bounding modulus of eigenvalues away from 1}
    Let $\lambda$ be a complex algebraic number whose minimal polynomial
    has degree at most $d$ and coefficients bounded in bitsize by $\tau$.
    Assume that $|\lambda| \neq 1$. 
    Then we have
    $
        \left||\lambda| - 1\right| 
        > 
        2^{-\left(\tau d\right)^{O(1)}}.
    $
\end{lemma}
\begin{proof}
    See Appendix \ref{Appendix: Mignotte}.
\end{proof}

%% file: bound-overview.tex
%!TEX root = main.tex

\section{Preliminaries}\label{Section: Preliminaries}
	
	\subsection{Converting the matrix to real Jordan normal form}
	
	To obtain a bound on the escape time it will be important to work with instances of the Escape Problem in real Jordan normal form.
	In the following, let $\mathbb{A}$ denote the field of algebraic numbers.
	We establish the following reduction to this case:
	
\begin{restatable}{lemma}{jordform}
\label{Lemma: converting to real JNF}
		Let $(K,A)$ be an instance of the Compact Escape Problem.
		Assume that $K$ is given by a formula involving 
		$s$ polynomial equations and equalities 
		$P \bowtie 0$
		where 
		$P \in \Z[x_1,\dots,x_n]$
		is a polynomial in $n$ variables 
		of degree at most $d$
		whose coefficients are bounded in bitsize by $\tau$.
		
		Let 
		$\gamma_1,\dots,\gamma_m \in \R$ 
		denote the real and imaginary parts of the eigenvalues of $A$.
		Let $\delta$ be a bound on the degrees of $\gamma_1,\dots,\gamma_m$.
		
		Then there exists an equivalent instance 
		$(J,K')$
		of the Compact Escape Problem 
		where
		$J \in \mathbb{A}^{(n + m) \times (n + m)}$ is in real Jordan normal form 
		and 
		$K'$ is given by a formula involving at most 
		$s + 3m$
		polynomial equations and equalities 
		$P \bowtie 0$
		where 
		$P \in \Z[x_1,\dots,x_{n + m}]$
		is a polynomial in $n + m$ variables 
		of degree at most $\delta \cdot d$
		whose coefficients are bounded in bitsize by $\tau + d (\log(2n) + \log(\delta + 1) + \sigma)$,
		where $\sigma$ depends polynomially on $n$ and the bitsize of the entries of $A$.
	\end{restatable}
	\begin{proof}
		See Appendix \ref{Appendix: converting to real JNF}.
	\end{proof}
	
\subsection{Decomposing \texorpdfstring{$K$}{K}}

Let $K \subseteq \R^{n}$ be a compact semialgebraic set.
Let $A \in \R^{n \times n}$ be a matrix in real Jordan normal form,
\[
    A = 
    \begin{pmatrix}
        J_1 &        & \\ 
            & \ddots & \\ 
            &        & J_m
    \end{pmatrix}.
\]
Here, each $J_i$ is a real Jordan block of the form 
\[
    J_i = \begin{pmatrix}
            \Lambda_{i} & I_{i} & & \\ 
                      & \ddots & \ddots & \\ 
                      &        &  \Lambda_{i}      & I_i    \\
                      &        &        & \Lambda_{i} 
          \end{pmatrix},
\]
where $\Lambda_{i,1}$ is either a real number or a $2\times 2$ real matrix of the form 
$\begin{pmatrix} a_i & -b_i \\ b_i & a_i \end{pmatrix}$
and, accordingly,
$I_i$ is either the real number $1$ or the $2\times 2$ identity matrix.
The elements $\Lambda_{i}$ correspond to real or complex eigenvalues $\lambda_i \in \C$ of $A$.
By slight abuse of language we call $|\lambda_i|$ the modulus of $\Lambda_i$.
By further slight abuse of language we define the ``eigenspace'' of $\Lambda_i$ as the one- or two-dimensional space 
spanned by the vectors that correspond to the first entry of the Jordan block $J_i$.
%\EN{Admittedly not a great formulation but the definition is extremely obvious and becomes annoying to state otherwise.}
The ``generalised eigenspaces'' for $\Lambda_i$ are defined analogously.

Write $\R^n$ as the direct sum of two spaces
$
    \R^n = V_{\operatorname{rec}} \oplus V_{\operatorname{non-rec}}
$
where $V_{\operatorname{rec}}$ is the direct sum of the eigenspaces for eigenvalues of modulus $1$,
and $V_{\operatorname{non-rec}}$ is the direct sum of the eigenspaces and generalised eigenspaces for eigenvalues of modulus $\neq 1$ and 
the generalised eigenspaces for eigenvalues of modulus $1$.
By convention, if $A$ has no eigenvalues of modulus $1$ we let $V_{\operatorname{rec}} = 0$.
Similarly, if $A$ has only eigenvalues of modulus $1$ and no generalised eigenvalues we let $V_{\operatorname{non-rec}} = 0$. 
Thus, we decompose the state space $\R^n$ into a part $V_{\operatorname{rec}}$ on which $A$ exhibits purely rotational behaviour,
and a part $V_{\operatorname{non-rec}}$ where $A$ is additionally expansive or contractive.

We will work with several different norms throughout this paper.
In addition to the familiar $\ell^2$ and $\ell^{\infty}$ norms we introduce a third norm, depending on the matrix $A$, that combines features of the two.
It facilitates block-wise arguments while ensuring that the restriction of $A$ to $V_{\operatorname{rec}}$ is an isometry.

Write $\R^n$ as a direct sum 
$
    \R^n = V_1 \oplus \dots \oplus V_s \oplus W_1 \oplus \dots \oplus W_t,
$
where $V_1,\dots,V_s$ correspond to the Jordan blocks of $A$ associated with real eigenvalues and 
$W_1,\dots,W_t$ correspond to the Jordan blocks of $A$ associated with non-real eigenvalues.
Let $\pi_{W_j} \colon \R^n \to W_j$ and $\pi_{V_j} \colon \R^n \to V_j$ denote the orthogonal projections 
onto $W_j$ and $V_j$ respectively.

For a vector $x \in V_i$, let 
$
    \Jnorm{x}^{V_i} = \norm{x}_\infty.
$
For a vector $x = (x_1,y_1,\dots,x_{k},y_k) \in W_i$, let
\[
    \Jnorm{x}^{W_i} = \max_{j = 1, \dots, k} \left( \sqrt{x_j^2 + y_j^2} \right).
\]
For a vector $x \in \R^n$, let 
\[
    \Jnorm{x} = \max\left\{
                        \max_{j = 1,\dots, s} \Jnorm{\pi_{V_j}(x)}^{V_j},
                        \max_{j = 1,\dots, t} \Jnorm{\pi_{W_j}(x)}^{W_j}
                    \right\}.
\]
Call $\Jnorm{x}$ the Jordan norm of $x$. 
Observe that $\Jnorm{x}$ depends on the choice of the $V_i$'s and $W_i$'s. 
The Jordan norm compares to the $\ell^2$- and $\ell^{\infty}$- norms as follows:
\[
    n^{-1/2} \Jnorm{x} \leq n^{-1/2} \norm{x}_2 \leq \norm{x}_{\infty} \leq \Jnorm{x} \leq \norm{x}_{2} \leq n^{1/2} \norm{x}_{\infty} \leq n^{1/2} \Jnorm{x}.
\]

Let $\varepsilon > 0$. 
Consider the ball $\JBall(0,\varepsilon) \subseteq \R^n$ about $0$ with respect to the distance induced by the $\Jnorm{\cdot}$-norm.
We partition $K$ into three sets: 
\begin{align*}
    &K_{\operatorname{rec}} = K \cap V_{\operatorname{rec}} \\ 
    &K_{< \varepsilon} = K \cap \left(V_{\operatorname{rec}} \oplus \left(\left(V_{\operatorname{non-rec}} \cap \JBall(0,\varepsilon)\right) \setminus \{0\}\right)\right)\\
    &K_{\geq \varepsilon} = K \cap \left(V_{\operatorname{rec}} \oplus \left(V_{\operatorname{non-rec}} \setminus \JBall(0,\varepsilon) \right)\right)
\end{align*}

%% file: kronecker-en.tex
\section{A quantitative version of Kronecker's theorem for complex algebraic numbers}

Our central tool for bounding the escape time in the recurrent case is a quantitative version of Kronecker's theorem for complex algebraic numbers.

Let $(\lambda_1,\dots,\lambda_m)$ be complex algebraic numbers of modulus $1$.
Our goal is to find for all $\varepsilon > 0$ a bound $N$ such that for all $(\alpha_1,\dots,\alpha_m) \in \mathbb{T}^m$
contained in the closure of the sequence $(\lambda_1^t,\dots,\lambda_m^t)_{t \in \N}$ there exists $t \leq N$ such that 
$|\lambda_j^t - \alpha_j| < \varepsilon$ for all $j = 1,\dots,m$.

We first consider the case where the $\lambda_j$'s do not admit any integer multiplicative relations.
In this case we can employ the following quantitative version of the continuous formulation of Kronecker's theorem, proved in \cite{QuantKroneckerSurvey}:

\begin{theorem}[{\cite[Theorem 4.1]{QuantKroneckerSurvey}}]\label{Theorem: Turan Kronecker}
    Let $\varphi_1,\dots,\varphi_N$ and $\zeta_1,\dots,\zeta_N$ be real numbers.
    Let $\varepsilon_1,\dots,\varepsilon_N$ be positive real numbers with 
    $\varepsilon_j < 1/2$ for all $j$.
    Let 
    $
        M_j = \left\lceil \tfrac{1}{\varepsilon_j} \log \tfrac{N}{\varepsilon_j}  \right\rceil.
    $    
    Let $\varphi = (\varphi_1,\dots,\varphi_N)$.
    Let 
    $
        \delta = \min \Set{|\varphi \cdot m|}{m \in \Z^N, |m_j| < M_j, m \neq 0}.
    $
    Assume that $\delta > 0$.
    Then in any interval $I$ of length $T \geq 4/\delta$ there is a real number $t$ such that 
    $
        \norm{\varphi_j t - \zeta_j} < \varepsilon_j,
    $
    where $\norm{\cdot}$ denotes distance to the nearest integer.
\end{theorem}

Intuitively, the number $\delta$ in Theorem \ref{Theorem: Turan Kronecker} is a quantitative measure of the linear independence of the $\varphi_j$'s,
as it bounds away from zero all integer linear combinations of the $\varphi_j$'s with suitably bounded coefficients. 
In our case we consider the numbers $\varphi_j = \log \lambda_j$.
For our purpose we need to obtain a bound on $t$, and thus a bound on $\delta$, in terms of the algebraic complexity of the numbers 
$\lambda_1,\dots,\lambda_m$.
This is achieved by invoking a quantitative version of Baker's theorem on linear forms in logarithms due to Baker and Wüstholz \cite{BakerWustholz}.
Recall that any algebraic number $\mu$ is the root of a unique irreducible polynomial $p_{\mu}$ with pairwise coprime integer coefficients.
The \emph{height} of an algebraic number $\mu$ is the maximum of the absolute values of the coefficients of $p_{\mu}$.
The \emph{degree} of $\mu$ is the degree of $p_{\mu}$.
Recall that a field $E$ is called an \emph{extension} of a field $F$ if $E$ contains $F$ as a subfield.
The \emph{degree} of a field extension $E \supseteq F$ is the dimension of $E$ as an $F$-vector space.

\begin{theorem}\label{Theorem: Baker}
	Let $\mu_1,\dots,\mu_N$ be algebraic numbers, none of which is equal to $0$ or $1$.
	Let 
	\[
	L(z_1,\dots,z_N) = b_1 z_1 + \dots + b_N z_N
	\]
	be a linear form with rational integer coefficients $b_1,\dots,b_N$.
	Let $B$ be an upper bound on the absolute values of the $b_j$'s.
	For $j = 1,\dots, N$, let $A_j\geq \exp(1)$
	be a bound on the height of $\mu_j$.
	Let $d$ be the degree of the field extension
    $\Q(\mu_1,\dots,\mu_N)$
    generated by $\mu_1,\dots,\mu_N$ over $\Q$.
	Fix a determination of the complex logarithm $\log$.
	Let 
	$
	\Lambda = L(\log \mu_1, \dots, \log \mu_N)
	$.
	If 
	$\Lambda \neq 0$
	then 
	\[
	\log |\Lambda|
	>
	-(16Nd)^{2(N + 2)}\log A_1 \cdot \dots \cdot \log A_N \log B.
	\]
\end{theorem}

Finally, in the case where the $\lambda_j$'s admit integer multiplicative relations, we employ Theorem \ref{Theorem: Masser} to bound their complexity. 
We arrive at the following result:

\begin{theorem}\label{theorem: General Quantitative Kronecker}
    Let $(\lambda_1,\dots,\lambda_m)$ be complex algebraic numbers of modulus $1$.
    Assume that the numbers 
    $2\pi i, \log \lambda_1,\dots, \log \lambda_s$ 
    are linearly independent over the rationals,
    where $0 \leq s \leq m$.
    Let $d$ be the degree of the field extension
    $\Q(\lambda_1,\dots,\lambda_s)$.
    Let 
    $A_1,\dots,A_s \geq \exp(1)$
    be upper bounds on the heights of $\lambda_1,\dots,\lambda_s$. 
    Let 
    $\ell \in \N$,
    and 
    $\varepsilon_{s + 1},\dots,\varepsilon_m \in \Z^s$
    be such that 
    \[
        \lambda_j^{\ell} = (\lambda_1,\dots,\lambda_s)^{\varepsilon_j}
    \]
    for all $j = s + 1,\dots, m$.
    By convention, if $s = 0$ the right-hand side of the above equation is to 
    be taken equal to $1$.

    Let 
    \[
        L 
        =
        \max\left\{\ell, \sum_{k = 1}^s |\varepsilon_{s+1, k}|,\dots,\sum_{k = 1}^s |\varepsilon_{m, k}|\right\}.
    \]

    Let 
    $\alpha_1,\dots,\alpha_m \in \mathbb{T}^m$ 
    be such that any rational linear relation between the numbers 
    $2\pi i, \log \lambda_1,\dots,\log \lambda_m$
    is also satisfied by the numbers 
    $2\pi i, \log \alpha_1,\dots,\log \alpha_m$.
    Let $\varepsilon > 0$.
    Then there exists a positive integer 
    \[
        t
        \leq 
            8\pi \ell 
            \left(
                \tfrac{2\pi L}{\varepsilon}
            \right)^s
            \left(
                2s
                \tfrac{2\pi L}{\varepsilon}
                \left\lceil
                    \tfrac{4\pi L}{\varepsilon}
                    \log \tfrac{4\pi s L}{\varepsilon}
                \right\rceil 
            \right)
            ^
            {
                \left(
                    16(s + 1)d
                \right)^{2(s + 3)}
                \log A_1 \cdot \dots \cdot  \log A_s
            }
        + \ell 
    \]
    such that 
    $
        \left|\lambda_j^t - \alpha_j\right| < \varepsilon
    $
    for $j = 1,\dots,m$.
\end{theorem}
\begin{proof}
    An outline of the proof is sketched above. See Appendix \ref{Appendix: Kronecker} for a full proof.
\end{proof}

For the purpose of bounding the escape time, the following coarse bound suffices:

\begin{corollary}\label{Corollary: qualitative Kronecker bound}
    There exists an integer function 
    $
        \operatorname{Kron}(n,\tau,P) \in 2^{\left(\tau P\right)^{n^{O(1)}}},
    $
    such that the following holds true:

    Let
    $\lambda_1,\dots,\lambda_n$
    be algebraic numbers of modulus $1$.
    Assume that the degree of each $\lambda_j$ is bounded by $n$.
    Let $\tau$ be a bound on the bitsize of the coefficients of the minimal polynomials of the $\lambda_j$'s.
    Let $P$ be a positive integer.
    Let $\alpha_1, \dots, \alpha_n$ be complex numbers which are contained in the closure of
    the sequence $(\lambda_1^t,\dots,\lambda_n^t)_{t \in \N}$.
    Then there exists a $t \leq \operatorname{Kron}(n,\tau,P)$ such that 
    $|\alpha_j - \lambda_j^t| < 2^{-P}$ for all $j \in \{1,\dots,n\}$. 
\end{corollary}
\begin{proof}
    By Kronecker's theorem, any integer multiplicative relation between the $\lambda_j$'s is also satisfied 
    by the $\alpha_j$'s.
    Theorem \ref{theorem: General Quantitative Kronecker} hence yields a bound on $t$ such that 
    $|\alpha_j - \lambda_j^t| < 2^{-P}$ holds for all $j \in \{1,\dots,n\}$.

    This bound is given in terms of quantities $s$, $d$, $\ell$, $\varepsilon_{s + 1},\dots,\varepsilon_m \in \Z^s$,
    $A_1,\dots,A_s$, and $L$.
    It remains to show that these quantities can be chosen to be suitably bounded in terms of $n$ and $\tau$.

    Proposition \ref{Proposition: integer multiplicative relations in normal form} in Appendix \ref{Appendix: Kronecker},
    which is mainly based on Theorem \ref{Theorem: Masser}, 
    shows that numbers $\ell$ and $\varepsilon_1,\dots,\varepsilon_m$ can be computed in polynomial space.
    In particular, the absolute size of $L$ and $\ell$ is of the form 
    $2^{(n\tau)^{O(1)}}$.
    The numbers $\log A_i$ are bounded by $\tau$ by assumption.
    We have $s \leq m \leq n$ by definition.
    Finally, we have assumed that each $\lambda_j$ has degree at most $n$.
    It follows that the degree $d$ of the field extension 
    $\Q(\lambda_1,\dots,\lambda_s)$
    is bounded by $n^n$.
    The result follows from Theorem \ref{theorem: General Quantitative Kronecker}.
\end{proof}

%% file: recurrent-eigenspace.tex
%!TeXroot = main.tex

\section{The recurrent eigenspace}
\label{sec:recspace}

The next lemma establishes as a special case an escape bound for all initial values $x \in K_{\operatorname{rec}}$.
In order to combine the recurrent and the non-recurrent case we need a stronger result, however.
Thus, we establish not only a bound on the escape time for all initial values $x \in K_{\operatorname{rec}}$, 
but a bound $N$ such that every $x \in V_{\operatorname{rec}}$ -- not just in $K_{\operatorname{rec}}$ -- 
has distance at least $1/N$ -- not just positive distance -- from $K$.
Further, note that Lemma \ref{Lemma: recurrent escape bound} is still applicable in the special cases where 
$K_{\operatorname{rec}} = \emptyset$ or $V_{\operatorname{rec}} = 0$.

\begin{lemma}\label{Lemma: recurrent escape bound}
    There exists an integer function 
    $
        \operatorname{Rec}(n, d, \tau) \in 2^{\left(\tau d\right)^{n^{O(1)}}}
    $
    with the following property: 

    Let $A \in \mathbb{A}^{n \times n}$ be a matrix in real Jordan normal form with algebraic entries.
    Assume that the minimal polynomial of $A$ has rational coefficients whose bitsize is bounded by $\tau$.
    Let $K \subseteq \R^n$ be a compact semialgebraic set of complexity at most $(n, d,\tau)$.
    If every point $x \in K_{\operatorname{rec}}$ escapes $K$ under iterations of $A$ then 
    for all $x \in V_{\operatorname{rec}}$ there exists 
    $t \leq \operatorname{Rec}(n,d,\tau)$
    such that 
    \[
        \operatorname{dist}_{\ell^2}(A^t x, K) > \frac{\sqrt{n}}{\operatorname{Rec}(n,d,\tau)}.
    \]
\end{lemma}
\begin{proof}
    The full proof is given in Appendix \ref{Appendix: Recurrent}.
    We only sketch an outline here.
    
    We first prove the result for initial points $x \in K_{\operatorname{rec}}$.
    For these points, the closure of the orbit $\clos{\O_A(x)}$ of $x$ under $A$ is a compact semialgebraic set.
    We employ Corollary \ref{Corollary: qualitative Kronecker bound} to obtain for all $\varepsilon > 0$ 
    a doubly exponential bound $N$ such that for all $x \in K_{\operatorname{rec}}$ and all 
    $y \in \clos{\O_A(x)}$ there exists $t \leq N$ such that 
    $\norm{A^t x - y}_2 < \varepsilon$.
    We then use Theorem \ref{Theorem: Jeronimo} to obtain a uniform at most doubly exponentially small lower bound on the quantity
    \[
        \inf_{x \in K_{\operatorname{rec}}} \sup_{y \in \clos{\O_A(x)}} \inf_{z \in K} \norm{y - z}_2^2.
    \]
    In order to apply this theorem we construct an auxiliary semialgebraic set, whose complexity is controlled 
    by Theorem \ref{Theorem: singly exponential quantifier elimination precise}.
    Combining these two steps, we obtain a function $\operatorname{Rec}_0$ that satisfies the statement of the lemma
    for all initial points $x \in K_{\operatorname{rec}}$.
    
    Finally, we extend the result to all initial points $x \in V_{\operatorname{rec}}$.
    The special case where $K_{\operatorname{rec}} = \emptyset$ is treated using Theorem \ref{Theorem: Jeronimo}.
    
    In the case where $K_{\operatorname{rec}}$ is non-empty we obtain from Lemma \ref{Lemma: compact semialgebraic set separation bound}
    that every $x \in V_{\operatorname{rec}}$ which is doubly exponentially close to $K$ with a sufficiently large constant in the third exponent 
    is already doubly exponentially close to $K_{\operatorname{rec}}$, with a slightly smaller constant in the third exponent.
    Now, any point that is sufficiently far away from $K$ trivially satisfies the claim.
    By the preceding discussion, points $x \in V_{\operatorname{rec}}$ that are sufficiently close to $K$ are already sufficiently close to $K_{\operatorname{rec}}$,
    so that there exists an escaping orbit $\clos{\O_A(x')}$ with $x' \in K_{\operatorname{rec}}$ which 
    is close to the orbit of $x$ since $A$ is an isometry on $V_{\operatorname{rec}}$.
    This allows us to reduce the result to the already established result for initial values in $K_{\operatorname{rec}}$.
\end{proof}

%% file: nonrecurrent-eigenspace.tex
%!TEX root = main.tex
\section{The non-recurrent eigenspace}
\label{sec:nonrecspace}

The next lemma concerns the subset $K_{\geq \varepsilon}$ of $K$ containing the points in $K$ that are 
bounded away from $V_{\operatorname{rec}}$ by some $\varepsilon > 0$.

For any such point, there exist coordinates (or pairs of coordinates if the corresponding eigenvalues are not real) whose contribution to the Jordan norm is greater than $\varepsilon$.
Moreover, the contribution to the Jordan norm of these coordinates does not stay constant under applications of $A$.
If the contribution to the norm of at least one such coordinate is increasing under applications of $A$,
%for instance if the associated eigenvalue is greater than 1, 
the orbit will eventually leave $K$, since $K$ is compact. 
Moreover, Theorem \ref{Theorem: Basu-Roy radius bound} yields an upper bound on the escape time.
% Moreover, the upper bound on the absolute values of the coordinates of all points in $K$ established in Theorem \ref{Theorem: Basu-Roy radius bound} yields an upper bound on the escape time.

Coordinates whose contribution to the norm is decreasing under applications of $A$ will, after sufficiently many iterations, contribute less than $\varepsilon$.
We establish a uniform upper bound on the number of iterations required to ensure this for all such coordinates.
Combining this with the previous bound, we obtain a number $N$ such that after at most $N$ applications of $A$, every $x \in K_{\geq \varepsilon}$ has either escaped $K$, entered $K_{< \varepsilon} \cup K_{\operatorname{rec}}$, 
or it remains in $K_{\geq \varepsilon}$ because it has a component whose contribution to the norm was initially smaller than $\varepsilon$, but grew beyond $\varepsilon$ under iteration of $A$.
In the last case, the point will grow in norm beyond the bound established in Theorem \ref{Theorem: Basu-Roy radius bound} and thus escape $K$ after a further $N$ applications of $A$.
This yields a uniform bound on the number of iterations that are required for any point $x \in K_{\geq \varepsilon}$ to either leave $K$ entirely or move into $K_{< \varepsilon} \cup K_{\operatorname{rec}}$.

The overall structure of this proof closely follows the one given in~\cite{CLOW20}, 
where the assumptions allow the authors to restrict the discussion to real eigenvalues.

\begin{lemma}\label{Lemma: non-recurrent overall bound}
	There exists an integer function 
	$
		\operatorname{NonRec}(n,d,\tau,P) \in 2^{\left(d\tau P\right)^{n^{O(1)}}}
	$
	with the following property: 

	Let $K$ be a compact semialgebraic set of complexity at most $(n,d,\tau)$.
	Let $A \in \mathbb{A}^{n \times n}$ be a matrix in real Jordan normal form.
	Assume that the characteristic polynomial of $A$ has rational coefficients whose bitsize is bounded by $\tau$.
	Let $P$ be a positive integer.

	Then for all $x \in K_{\geq 2^{-P}}$ there exists 
	$t \leq \operatorname{NonRec}(n,d,\tau,P)$
	such that 
	$A^t x \notin K_{\geq 2^{-P}}$.
\end{lemma}
\begin{proof}
	See Appendix \ref{Appendix: non-recurrent} for details.
\end{proof}

%% file: combined-bound.tex
%!TEX root = main.tex
\section{Proof of Theorem \ref{Theorem: main theorem}}\label{Section: proof of main theorem}

In the previous two sections, we successively showed how to establish a bound on the escape 
time for an instance $(A,K)$ when the orbit remains in the recurrent eigenspace and how the 
orbit behaves when it 
starts away from the recurrent eigenspace. In this section, we show how to combine both results
in order to establish an escape bound for any starting point in $K$. This will thus
prove Theorem \ref{Theorem: main theorem}.

Let $(A_0,K_0)$ be an instance of the compact escape problem,
where $K_0 \subseteq \R^n$ is a compact semialgebraic set of complexity at most $(n_0,d_0,\tau_0)$
and $A_0 \in \Q^{n \times n}$ is a square matrix with rational entries whose bitsize is bounded by $\tau_0$.
Assume that every point $x \in K_0$ escapes $K_0$ under iterations of $A_0$.

Apply Lemma \ref{Lemma: converting to real JNF} to convert the instance $(A_0, K_0)$ into an equivalent 
instance $(A,K)$ such that $A \in \mathbb{A}^{n \times n}$ is in real Jordan normal form.
Then the set $K$ has complexity at most $(n,d,\tau)$, 
were $n = 2n_0$, $d = n_0 d_0$, and 
$\tau = (n_0 \tau_0 d_0)^{C_{\tau}}$
for some absolute constant $C_{\tau}$.
By construction, the characteristic polynomial of $A$ has rational coefficients of bitsize at most $\tau$.

% We start by considering the case where $V_{\operatorname{rec}}$ is the emptyset (\emph{i.e.} 
% there is no eigenvalue of modulus 1 in $A$). By \autoref{Lemma: Jeronimo}, there
% exists a bound $\varepsilon_{\operatorname{rec}}$ that is doubly exponential in the number of
% dimensions and simply exponential in the bitsize and degree of polynomials such that 
% the distance with respect to the $l_2$ norm ($\operatorname{dist}_{l_2}$) between the $0$ 
% vector and $K$ is at least $\varepsilon_{\operatorname{rec}}$.
% If $V_{\operatorname{rec}}$ is not empty we set $\varepsilon_{\operatorname{rec}}=1$ by 
% default.

Let $\operatorname{Rec}$ be the function from Lemma \ref{Lemma: recurrent escape bound}.
Let 
$\varepsilon = \frac{1}{\operatorname{Rec}(n,d,\tau)}$
and $N_{\operatorname{rec}} = \operatorname{Rec}(n,d,\tau)$.
Let $x \in K$. 
If $x \in K_{\operatorname{rec}}$ then $x$ escapes within $N_{\operatorname{rec}}$ steps.
Suppose that $x \in K_{< \varepsilon}$.

Then there are two possibilities: 
\begin{enumerate}
    \item 
        We have $A^{t} x \notin K_{\geq \varepsilon}$ for all $t \leq N_{\operatorname{rec}}$.
    \item
        We have $A^t x \in K_{\geq \varepsilon}$ for at least one $t \leq N_{\operatorname{rec}}$.
\end{enumerate}
In the first case, the orbit of $x$ remains close to $V_{\operatorname{rec}}$
for long enough that we can rely on Lemma \ref{Lemma: recurrent escape bound}.
Indeed, let $x_0$ denote the orthogonal projection of $x$ onto $V_{\operatorname{rec}}$.
Let $t\leq N_{\operatorname{rec}}$ be such that $\operatorname{dist}_{\ell^2}(A^t x_0, K) > \sqrt{n}\varepsilon$.
Since $A^{t} x \notin K_{\geq \varepsilon}$, we have 
$
    \Jnorm{A^t x - A^t x_0} < \varepsilon,
$
so that 
$
    \norm{A^t x - A^t x_0}_2 < \sqrt{n} \varepsilon.
$
Let $y \in K$.
Then 
\[
    \norm{A^t x - y}_2 
    \geq  
    \norm{A^t x_0 - y}_2 
    - 
    \norm{A^t x - A^t x_0}_2
    >
    \sqrt{n}\varepsilon
    -
    \sqrt{n}\varepsilon
    = 
    0.
\]
Thus, $x$ escapes $K$ under iterations of $A$. 

In the second case, let $t_1$ be such that $A^{t_1} x \in K_{\geq \varepsilon}$.
Let $\operatorname{NonRec}$ be the function from Lemma \ref{Lemma: non-recurrent overall bound}.
Let $N_{\geq \varepsilon} = \operatorname{NonRec}(n,d,\tau,\left\lceil\log (1/\varepsilon)\right\rceil)$.
By Lemma \ref{Lemma: non-recurrent overall bound} there exists $t_2 \leq N_{\geq \varepsilon}$ such 
that $A^{t_2} A^{t_1} x$ 
is contained either in 
$K_{< \varepsilon} \cup K_{\operatorname{rec}}$ 
or in the complement of $K$.
In the latter case we are done.
In the former case we apply the initial case distinction: 
either for all $t \leq N_{\operatorname{rec}}$ we have 
$A^t A^{t_2} A^{t_1} x \notin K_{\geq \varepsilon}$
or we have 
$A^{t_3} A^{t_2} A^{t_1} x \in K_{\geq \varepsilon}$
for at least one $t_3 \leq N_{\operatorname{rec}}$.
Once again, in the first case, the point has escaped.
By repeating this reasoning, we 
construct a (finite or infinite) sequence $t_1,t_2,\dots$ 
such that $t_i \leq N_{\operatorname{rec}}$ if $i$ is odd and $t_i \leq N_{\geq \varepsilon}$ if $i$ is even
and 
\[
    A^{t_s} \cdot \dots \cdot A^{t_1} x \in 
        \begin{cases}
            K_{< \varepsilon} \cup K_{\operatorname{rec}}    &\text{if }s \text{ is even,}\\ 
            K_{\geq \varepsilon} &\text{if }s \text{ is odd.}
        \end{cases}
\]
We claim that the sequence $t_1,t_2,\dots$ is finite and contains at most $n^3$ elements.

    Consider a real Jordan block 
%    \[
%        J
%        =
%        \begin{pmatrix}
%            \Lambda & I      &       &     \\ 
%                    & \ddots &\ddots &     \\ 
%                    &        & \Lambda & I \\
%                    &        &         & \Lambda 
%        \end{pmatrix}
%    \]
    of $A$ of size $m \leq n$ associated to the eigenvalue $\Lambda$.
    Denote by $x_J$ the orthogonal projection of $x$ onto the dimensions associated with this block.
    
    Assume first that $\Lambda$ is a real eigenvalue (as opposed to a $2\times 2$ block 
    representing a complex eigenvalue). 
    If $\Lambda = 0$, then clearly 
    $\Jnorm{J^k x_J}$ 
    is monotonically decreasing.
    Thus, assume in the sequel that $\Lambda \neq 0$.  
      
    Let $j \in \{1,\dots,m\}$.
    The $m - j + 1$'th component of the vector
    $J^k x_{J}$,
    viewed as a function of $t$, is an exponential polynomial
    $
        E_j(t) = \Lambda^t P(t),
    $
    where $P \in \R[z]$ is a real polynomial of degree $j - 1$.
    Consider the real function 
    \[
        (E_j(\cdot))^2 \colon \R \to \R, 
        \;
        (E_j(t))^2 = 
            \left|\Lambda\right|^{2t} |P(t)|^2.
    \]
    This function is differentiable in $t$ with derivative
    \[
        \tfrac{d}{dt} (E_j(t))^2
        =
        \Lambda^{2t}
        \left(
        \log (\Lambda^2) (P(t)^2)
        +
        2 P(t) P'(t) 
        \right).
    \]
    This derivative vanishes if and only if the factor 
    $\left(
        \log (\Lambda^2) (P(t)^2)
        +
        2 P(t) P'(t)
        \right)$ 
    vanishes.
    This factor is a polynomial of degree $2j - 2$, so that it has at most $2j - 2$ real zeroes.
    It follows that there exist numbers $t_{j,1},\dots,t_{j,m_j}$ with $m_j \leq 2j - 2$ such that the function 
    $(E_j(t))^2 - \varepsilon^2$ 
    does not change its sign in any of the open intervals 
    \[ 
        (0, t_{j,1}), 
        (t_{j,1},t_{j,2}),
        \dots,
        (t_{j,m_{j} - 1}, t_{j,m_j}),
        (t_{j,m_j}, +\infty).
    \]
    
    Thus, the norm $\Jnorm{J^t x_J}$ changes from smaller than $\varepsilon$ to bigger than $\varepsilon$ at most
    \[
        \sum_{j = 1}^m (2j - 2)
        =
        2\sum_{j = 1}^m j - 2m
        =
        (m + 1)m - 2m 
        =
        m^2 - m
    \]
    times.

The case where $\Lambda$ represents a complex eigenvalue $\lambda$ 
%\footnote{$\overline{z}$ represents the complex conjugate of $z$} 
is similar.
However, we now consider the evolution of the two coordinates corresponding to one $\Lambda$-block simultaneously.

For $j \in \{1,\dots,m\}$, write $E_j(t)$ for the $m - j + 1$'th component of the vector
$J^t x_J$, viewed as a function of $t$.
We have for all $j \in \{1,\dots,m/2\}$ that the function 
\[
F_j(t) = (E_{2j}(t))^2 +(E_{2j-1}(t))^2
\]
is an exponential polynomial
$
  F_j(t) = |\lambda|^t P_j(t),
$
where $P_j \in \R[z]$ is a real polynomial of degree $j - 1$.
Therefore, exactly as in case where $\Lambda$ is a real eigenvalue, 
the derivative of $F_j$ vanishes at most $2j-2$ times. From which we can deduce that the norm
$\Jnorm{J^tx_J}$ crosses the $\varepsilon$-threshold at most $m^2-m$ times.

Estimating generously, we have at most $n$ Jordan blocks of size at most $n$, each of which crosses the $\varepsilon$-threshold at most 
$n^2 - n$ times.
In total, we cross the threshold at most $n^3 - n^2$ times.
The total escape bound is hence 
$
    n^3 \max\{N_{\operatorname{rec}}, N_{\geq \varepsilon}\}.
$
By the same argument, the same escape bound holds true when the initial point $x$ lies in $K_{\geq \varepsilon}$.

Substituting the constants $N_{\operatorname{rec}}$, $N_{\geq \varepsilon}$, $n$, $d$, and $\tau$ with their definitions, we obtain the upper bound
\begin{align*}
    \operatorname{CompactEscape}&(n_0, d_0, \tau_0) =  \\
        (2n_0)^3 &\max\big\{\operatorname{Rec}\left(2n_0, n_0d_0, \left(n_0 d_0 \tau_0\right)^{C_{\tau}}\right), \\
                &\hspace{3em} \operatorname{NonRec}\left(2n_0, n_0d_0, \left(n_0 d_0 \tau_0\right)^{C_{\tau}}, 
                     \log \left\lceil\operatorname{Rec}\left(2n_0, n_0d_0, \left(n_0 d_0 \tau_0\right)^{C_{\tau}}\right) \right\rceil\right) \big\}.
\end{align*}
One easily verifies that $\operatorname{CompactEscape}(n,d,\tau) \in 2^{\left(d\tau\right)^{n^{O(1)}}}$ as claimed.

%% file: double-exp-example.tex
\section{A matching lower bound on escape time}\label{Section: Lower Bounds}

In Theorem \ref{Theorem: main theorem} we established a uniform upper bound on the escape time for all positive instances of the Compact Escape Problem.
Our bound is doubly exponential in the ambient dimension and singly exponential in the rest of the data.
We will now show that this bound cannot be significantly improved by 
showing that a doubly exponential bound cannot be avoided even for purely rotational systems. A second example displaying a doubly exponential lower bound is presented in 
Appendix~\ref{Appendix: example}.

\begin{example}
For $(n,d,\tau) \in \N^3$, let $K_{(n,d,\tau)} \subseteq \R^{n + 2}$
be the set of all points $(x,y,u_1,\dots,u_n)$ satisfying the (in)equalities:
$x^2 + y^2 = 1,u_1 = 2^{-\tau}$, $(x - 1)^2 + y^2 \geq u_n$ and for $1\leq i \leq n-1, \; u_{i + 1} = (u_i)^d$.

Hence, 
$	K_{(n,d,\tau)} = \left(S^1 \setminus B\left((1,0), 2^{-\tau d^{n - 1}}\right)\right) \times \left\{\left(2^{-\tau}, 2^{-\tau d},\dots, 2^{-\tau d^{n - 1}}\right) \right\}$,
where $S^1 \subseteq \R^2$ is the unit circle.
Let $a = \tfrac{3}{5}$, $b = \tfrac{4}{5}$. 
Let 
\[
A_{(n,d,\tau)} = 
	\begin{pmatrix}
						a & -b & 0  \\
						b & a &  0 \\ 
						0 & 0  & I_{n}
	\end{pmatrix}
\]
where $I_n$ is the $n\times n$- identity matrix.
It is easy to see that the complex number $\tfrac{3}{5} + i \tfrac{4}{5}$ has modulus $1$ and is not a root of unity.
It follows from Dirichlet's theorem on simultaneous Diophantine approximation that the orbit of $A$ 
is equal to $S^1 \times \left\{\left(2^{-\tau}, 2^{-\tau d},\dots, 2^{-\tau d^{n - 1}}\right) \right\}$, 
so that every initial point escapes under $A$. 

We claim that there exists a point $x \in K_{(n,d,\tau)}$ that requires $2^{\tau d^{n - 1}}$ steps to escape.
Indeed, let $x_0 \in K_{(n,d,\tau)}$ be an arbitrary initial point. 
Consider the orbit $x_t = A^t x_0$.
Let $N < 2^{\tau d^{n - 1}}$.
By the pigeonhole principle, the finite set of points $x_0,\dots,x_{N}$ contains at least one consecutive pair of points $x_i$, $x_j$ on the circle such that
the points $x_i$ and $x_j$ are joined by an arc of the circle of length strictly greater than $2/N$. 
It follows that we can ensure that none of the points $x_1,\dots,x_N$ is outside of $K_{(n,d,\tau)}$ by applying a suitable planar rotation to all points.
Since all planar rotations commute, there exists for each angle $\theta$ an initial point 
$x_{\theta} \in S^1 \times \left\{\left(2^{-\tau}, 2^{-\tau d},\dots, 2^{-\tau d^{n - 1}}\right) \right\}$,
such that the orbit of $x_{\theta}$ under $A$ is equal to the orbit of $x_0$ under $A$ rotated by $\theta$.
This proves the claim.

		\end{example}

%% file: app-real-jordan.tex
\appendix
\section{Computing real JNF in polynomial time}
\label{realjordan}

Given a matrix $A$ with rational entries, we discuss how to compute the real Jordan normal form $J$ of $A$ and the associated change of basis matrix $Q$ in polynomial time. 
First compute, in polynomial time, the (complex) Jordan normal form $J'$ and change of basis matrix $T$ such that $A = T J' T^{-1}$
using the algorithm from \cite{Cai94}. 

\proofsubparagraph{Computing $J$:} Suppose, without loss of generality, that
\[J' = \operatorname{diag}(J'_1, J'_2, \ldots, J'_{2k-1}, J'_{2k}, J'_{2k+1}, \ldots, J'_{2k+z})\]
where for $1 \leq j \leq k$, the Jordan blocks $J'_{2j-1}$ and $J'_{2j}$ have the same dimension and have conjugate eigenvalues $\lambda_j = a_j + b_ji$ and $\overline{\lambda} = a_j - b_ji$, respectively. The blocks $J'_{2k+1}, \ldots, J'_{2k+z}$, on the other hand, have real eigenvalues. $J$ is obtained by replacing, for each $1 \leq j \leq k$, $\operatorname{diag}(J'_{2j-1}, J'_{2j})$ with a real Jordan block of the same dimension with $\Lambda=\begin{bmatrix}
	a & -b \\
	b & a
\end{bmatrix}$
and keeping the blocks $J'_{2k+1}, \ldots, J'_{2k+z}$ unchanged.

\proofsubparagraph{Computing $Q$:}
Let $\kappa(j)$ denote the multiplicity of the Jordan block $J'_i$ for $1 \leq i \leq 2k+z$, and $v_1^1, \ldots, v_{\kappa(1)}^1, \ldots, v^{2k}_1, \ldots, v^{2k}_{\kappa(2k)}, \ldots,v^{2k+z}_1, \ldots, v^{2k+z}_{\kappa(2k+z)} \in \C^{m}$ be the columns of $T$. It will be the case that for all $1 \leq j \leq k$ and $l$, $v^{2j-1}_l = \overline{v^{2j}_l}$ in the sense that
$
v^{2j-1}_l = x^j_l + y^j_l i
$
and
$
v^{2j}_l = x^j_l - y^j_l i
$ for vectors $x^j_l, y^j_l \in \R^m$. Moreover, for $j > 2k$, $v^{2j}_l \in \R^m$.  Finally, columns of $Q$ are obtained from columns of $T$ as follows. For $1 \leq j \leq k$ and all $l$, replace $v^{2j-1}_l$ with $x^j_l$ and $v^{2j}_l$ with $y^j_l$ and keep $v^{2k+z}_l$ for all $l$ and $m > 0$ unchanged, in the same way the proof of existence of real Jordan normal form proceeds.

%\proofsubparagraph{Computing $Q^{-1}$:} Summarizing the construction above, $Q$ is obtained from $T$ by replacing columns $x+yi$ and $x-yi$, $x, y \in \R^{m}$ by $x$ and $y$, respectively. Since $x = \frac{1}{2}(x+yi) + \frac{1}{2}(x-yi)$ and $y = -\frac{1}{2}i (x+yi) + \frac{1}{2}i (x-yi)$, this construction is linear and we can write $Q = T \cdot A$ for some $A \in \C^{m\times m}$ with entries in $\{\frac{1}{2}, -\frac{1}{2}, \frac{1}{2}i, -\frac{1}{2}i, 1, 0\}$. Moreover, the linear transformation is clearly invertible: $x + yi = 1\cdot x + i y$ and $x-yi = 1\cdot x - (-i)y$, and hence $A^{-1} \in \C^{m\times m}$ with entries in $\{1, i, -i\}$. Finally, compute $Q^{-1}$ via $Q = T\cdot A \implies Q^{-1} = A^{-1} \cdot T^{-1}$, observing that we already know how to compute $T^{-1}$ in polynomial time.

%% file: app-JNF.tex
\section{Proof of Lemma \ref{Lemma: converting to real JNF}}
\label{Appendix: converting to real JNF}

\jordform*

By Appendix 
\ref{realjordan}
we can compute in polynomial time real algebraic numbers 
$\gamma_1,\dots,\gamma_m$
and a matrix 
$Q \in \Q(\gamma_1,\dots,\gamma_m)^{n \times n}$ 
such that 
$A = Q J Q^{-1}$,
where $J$ is in real Jordan normal form.

More precisely, we can compute in polynomial time:
\begin{enumerate}
    \item 
        Univariate polynomials with integer coefficients
        $f_1, \dots, f_m$
        such that 
        $f_j(\gamma_j) = 0$
        for all $j = 1,\dots,m$.
    \item 
        Rational numbers 
        $a_1, b_1,\dots,a_{m}, b_{m}$,
        such that 
        $\gamma_j$ 
        is the unique root of $f_j$ in the interval
        $[a_j,b_j]$..
    \item 
        For 
        $i = 1,\dots,n$ 
        and 
        $j = 1,\dots,n$
        polynomials of degree at most $\delta$
        $Q_{i,j} \in \Q[x]$
        and indexes 
        $\ell_{i,j}$ 
        such that the matrix 
        $Q$ 
        at row $i$ and column $j$
        is given by the algebraic number
        $Q_{i,j}(\gamma_{\ell_{i,j}})$.
\end{enumerate}

Let $\sigma \in \N$ be a common bound on the following quantities:

\begin{enumerate}
    \item The bitsize of the coefficients of $f_1,\dots,f_m$.
    \item The bitsize of the endpoints of the isolating intervals $[a_j,b_j]$.
    \item The bitsize of the coefficients of the polynomials $Q_{j,k}$.
\end{enumerate}

Then $\sigma$ is computable in polynomial time from $A$, so that it depends polynomially
on $n$ and the bitsize of the entries of $A$.

We fix 
$K' = (Q^{-1} K) \times \{\gamma_1,\dots,\gamma_m\}$\footnote{The last $m$ coordinates are added to allow us to manipulate these constants within the description of $K'$ through a formula.}
and note that 
$
    A^k x \in K
$
if and only if 
\[ 
    (J \times I_m)^k (Q^{-1} x, (\gamma_1,\dots,\gamma_m)) \in (Q^{-1} K) \times \{(\gamma_1,\dots,\gamma_m)\}.
\]
Thus, it remains to show that there exists a description of the set $K'$
with the claimed complexity.

Let 
$\Phi(x_1,\dots,x_n)$ 
be the formula that describes $K$.
We introduce fresh variables $z_1,\dots,z_m$ and consider the formula 
\[
    \Psi(z_1,\dots,z_m)
    \land 
    \widehat{\Phi}(x_1,\dots,x_n,z_1,\dots,z_m)
\]
where 
$\Psi(z_1,\dots,z_m)$ 
is the conjunction of the terms 
\[
    f_j(z_j) = 0 
    \land 
    z_j \geq a_j
    \land 
    z_j \leq b_j
\] 
which ensures $z_j=\gamma_j$ for $j = 1,\dots,m$, 
and 
$\widehat{\Phi}$
is obtained from 
$\Phi$
by replacing each atom 
$
    P(x_1,\dots,x_n) \bowtie 0  
$
in 
$\Phi$ 
by the atom 
\[
    P
    \left( 
        \sum_{k = 1}^n Q_{1,k}(z_{\ell_{1,k}}) x_k,
        \dots,
        \sum_{k = 1}^n Q_{n,k}(z_{\ell_{n,k}}) x_k,
    \right)
    \bowtie 0.
\]

It is not hard to see that this new formula describes the set $K'$.
Evidently, the number of variables in this description is $n + m$.
The formula $\Psi$ involves $3m$ polynomials of degree at most $\delta$ whose coefficients are bounded in bitsize by $\sigma$.

It remains to determine the complexity of the formula $\widehat{\Phi}$.
We claim that the degrees of the polynomials in $\widehat{\Phi}$ are bounded by 
$\delta \cdot d$ 
and that the bitsize of their coefficients is bounded by 
$\tau + d (\log(n) + \log(\delta + 1) + \sigma)$.
This is established by a straightforward but cumbersome calculation.
We recall the multinomial theorem:

\begin{lemma}[Multinomial theorem]
    Let $R$ be a ring.
    Let $N$ be a positive integer.
    Let $z_1,\dots,z_N \in R$.
    Then 
    \[
        \left(
            \sum_{k = 1}^N z_k
        \right)
        ^
        {
            e
        }
        =
        \sum_{j_1 + \dots + j_N = e}
        \binom{e}{j_1,\dots,j_N}
        \prod_{t = 1}^N x_t^{j_t},
    \]
    where 
    \[
        \binom{e}{j_1,\dots,j_N}
        =
        \frac{e!}
             {j_1! \cdot \dots \cdot j_N!}.
    \]
\end{lemma}

It will be convenient to make use of the following straightforward application of the distributivity law of multiplication over addition:

\begin{lemma}\label{Lemma: swapping sum and product}
    Let $I$ be a finite set.
    Let $J$ be a set-valued function that sends 
    each $i \in I$ to a finite set $J(i)$.
    Let $R$ be a ring. 
    For all $(i,j) \in I \times \coprod_{i \in I} J(i)$,
    let $a_{i,j}$ be an element of $R$.
    Then we have 
    \[
        \prod_{i \in I} \sum_{j \in J(i)}
            a_{i,j}
        =
        \sum_{f} \prod_{i \in I} a_{i,f(i)}
    \]
    where $f$ ranges over all functions
    \[
        f \colon I \to \coprod_{i \in I} J(i)
    \]
    satisfying $f(i) \in J(i)$ for all $i \in I$.
\end{lemma}

Now, let 
$P(x_1,\dots,x_n) \bowtie 0$
be an atom in $\Phi$.
This atom is a sum of monomials
\[
    C \cdot x_1^{e_1}\cdot \dots \cdot x_{n}^{e_{n}}
\]
with 
$\log |C| \leq \tau$ 
and $e_1 + \dots + e_n \leq d$.
It suffices to bound the degrees and the bitsize of the coefficients
of the polynomials that are obtained by applying our substitution
of variables to monomials of this form.

Under our substitution such a monomial becomes:
\[
    C
    \cdot 
    \left(\sum_{j = 1}^{n} Q_{1,j}(z_{\ell_{1,j}}) x_j\right)^{e_1}
    \cdot 
        \dots
    \cdot 
    \left(\sum_{j = 1}^{n} Q_{n,j}(z_{\ell_{n,j}}) x_j\right)^{e_{n}}
    =
    C
    \prod_{k = 1}^n
    \left(\sum_{j = 1}^{n} Q_{k,j}(z_{\ell_{k,j}}) x_j\right)^{e_k}
    .
\]
Apply the multinomial theorem to the expressions 
$\left(\sum_{j = 1}^{n} Q_{k,j}(z_{\ell_{k,j}}) x_j\right)^{e_k}$
to obtain:
\[
    C
    \cdot
    \prod_{k = 1}^{n}
    \left(
        \sum_{j_{k,1} + \dots + j_{k,n} = e_k}
        \binom{e_k}{j_{k,1},\dots,j_{k,n}}
        \prod_{t = 1}^{n} \left(Q_{k,t}(z_{\ell_{k,t}}) x_t\right)^{j_{k,t}}
    \right)
\]
Write
\[
    Q_{k,t}(z_{\ell_{k,t}})
    =
    \sum_{p = 0}^{\delta} \alpha_{k,t,p} z_{\ell_{k,t}}^{p}.
\]
Applying the multinomial theorem to the terms 
\[ 
    \left(Q_{k,t}(z_{\ell_{k,t}}) x_t\right)^{j_{k,t}}
    =
    \left(
        \sum_{p = 0}^{\delta} \alpha_{k,t,p} z_{\ell_{k,t}}^{p} x_t
    \right)^{j_{k,t}}
\]
we obtain
\[
    \left(Q_{k,t}(z_{\ell_{k,t}}) x_t\right)^{j_{k,t}}
    =
    \sum_{r_0 + \dots + r_{\delta} = j_{k,t}} 
    \binom{j_{k,t}}{r_0,\dots,r_{\delta}}
    \prod_{s = 0}^{\delta} 
        \alpha_{k,t,s}^{r_s} z_{\ell_{k,t,s}}^{s r_s} x_t^{r_s}.
\]
The full expression is hence:
\[
    C
    \cdot
    \prod_{k = 1}^{n}
    \left(
        \sum_{j_{k,1} + \dots + j_{k,n} = e_k}
        \binom{e_k}{j_{k,1},\dots,j_{k,n}}
        \prod_{t = 1}^{n} 
        \sum_{r_0 + \dots + r_{\delta} = j_{k,t}} 
        \binom{j_{k,t}}{r_0,\dots,r_{\delta}}
        \prod_{s = 0}^{\delta} 
            \alpha_{k,t,s}^{r_s} z_{\ell_{k,t,s}}^{s r_s} x_t^{r_s}
    \right).
\]
Write this as:
\[
    C
    \cdot
    \prod_{k = 1}^{n}
    \sum_{j_{k,1} + \dots + j_{k,n} = e_k}
    \binom{e_k}{j_{k,1},\dots,j_{k,n}}
    \prod_{t = 1}^{n} 
    \sum_{r_0 + \dots + r_{\delta} = j_{k,t}} 
        c_{k, j_{k,1},\dots,j_{k,n}, t, r_0,\dots,r_{\delta}}.
\]
Apply Lemma \ref{Lemma: swapping sum and product} to move out the innermost sum, thus obtaining an equal expression:
\[
    C
    \cdot
    \prod_{k = 1}^{n}
    \left(
        \sum_{j_{k,1} + \dots + j_{k,n} = e_k}
        \sum_{f} 
        \binom{e_k}{j_{k,1},\dots,j_{k,n}}
        \prod_{t = 1}^{n} 
            c_{k, j_{k,1},\dots,j_{k,n}, t, f(t)}
    \right).
\]
where the sum $\sum_f$ ranges over all functions 
$
    f \colon \{1,\dots,n\} \to \N^{\delta}
$
with 
$f(t) = (r_0,\dots,r_{\delta})$
satisfying
$r_0 + \dots + r_{\delta} = j_{k,t}$.

Write the result as:
\[
    C
    \cdot
    \prod_{k = 1}^{n}
    \sum_{j_{k,1} + \dots + j_{k,n} = e_k}
    \sum_{f}
    d_{k, j_{k,1}, \dots, j_{k, n}, f}.
\]
Apply Lemma \ref{Lemma: swapping sum and product} again to obtain that this is equal to:
\[
    \sum_{g}
    C 
    \cdot 
    \prod_{k = 1}^{n}
    d_{k, g(k)},
\]
where $g$ ranges over all functions
$
    g \colon \{1,\dots,n\} \to \N^{n} \times (\N^{\delta})^{\{1,\dots,n\}}
$
with 
$g(k) = (j_{k,1},\dots,j_{k,n}, f)$
satisfying 
$j_{k,1} + \dots + j_{k,n} = e_k$
and 
$f$
as above.

Thus, the final result is a sum of monomials of the form 
\begin{align*}
    C \cdot \prod_{k = 1}^{n} d_{k, g(k)}
    &=
    C \cdot 
    \prod_{k = 1}^{n}
    \binom{e_k}{j_{k,1},\dots,j_{k,n}}
    \prod_{t = 1}^{n} 
        c_{k, j_{k,1},\dots,j_{k,n}, t, f(t)}\\
    &= 
    C \cdot 
    \prod_{k = 1}^{n}
    \binom{e_k}{j_{k,1},\dots,j_{k,n}}
    \prod_{t = 1}^{n} 
    \binom{j_{k,t}}{r_0(t),\dots,r_{\delta}(t)}
        \prod_{s = 0}^{\delta} 
            \alpha_{k,t,s}^{r_s(t)} z_{\ell_{k,t,s}}^{s r_s(t)} x_t^{r_s(t)},
\end{align*}
Where 
$j_{k,1} + \dots + j_{k,n} = e_k$
and 
$r_0(t),\dots,r_\delta(t)$
are functions of $t$ satisfying 
$r_0(t) + \dots + r_{\delta}(t) = j_{k,t}$.

The degrees of these monomials are bounded by 
$
    \delta \cdot d
$.

Let us compute a bound on the bitsize of the coefficients.
We have:
\begin{align*}
    &\log 
    \left(
        |C| \cdot 
    \prod_{k = 1}^{n}
    \binom{e_k}{j_{k,1},\dots,j_{k,n}}
    \prod_{t = 1}^{n} 
    \binom{j_{k,t}}{r_0(t),\dots,r_{\delta}(t)}
        \prod_{s = 0}^{\delta} 
            \left|\alpha_{k,t,s}\right|^{r_s(t)}
    \right)\\
    &\leq 
    \tau 
    +
    \sum_{k = 1}^{n} 
        \log
                \binom{e_k}{j_{k,1},\dots,j_{k,n}}
    + \sum_{k = 1}^{n}
        \sum_{t = 1}^{n} 
            \log 
                \binom{j_{k,t}}{r_0(t),\dots,r_{\delta}(t)}
    + \sum_{k = 1}^{n}
        \sum_{t = 1}^{n} 
            \sum_{s = 0}^{\delta} 
                r_s(t)
                \sigma.
\end{align*}

Use the estimate $\binom{f}{k_1,\dots,k_m} \leq m^f$ to obtain:
\begin{align*}
    &\log 
    \left(
        |C| \cdot 
    \prod_{k = 1}^{n}
    \binom{e_k}{j_{k,1},\dots,j_{k,n}}
    \prod_{t = 1}^{n} 
    \binom{j_{k,t}}{r_0(t),\dots,r_{\delta}(t)}
        \prod_{s = 0}^{\delta} 
            \left|\alpha_{k,t,s}\right|^{r_s(t)}
    \right)\\
    &\leq 
    \tau 
    +
    \sum_{k = 1}^{n} 
        e_k
        \log (n)
    + \sum_{k = 1}^{n}
        \sum_{t = 1}^{n} 
            j_{k,t}
            \log (\delta + 1)
    + \sum_{k = 1}^{n}
        \sum_{t = 1}^{n} 
            \sum_{s = 0}^{\delta} 
                r_s(t)
                \sigma.\\
    &\leq
    \tau 
    +
    d
        \log (n)
    +
    d
            \log (\delta + 1)
    + 
    d
                \sigma.\\ 
    &=
    \tau + d (\log(n) + \log(\delta + 1) + \sigma).
\end{align*}
Thus, everything is shown.

%% file: app-Mignotte.tex
\section{Proof of Lemma \ref{Lemma: bounding modulus of eigenvalues away from 1}}\label{Appendix: Mignotte}

Recall the following classic theorem due to Mignotte \cite{Mig82}.

\begin{theorem}[Mignotte]\label{Theorem: Mignotte bound}
    Let $P \in \Z[x]$ be a square-free univariate polynomial of degree at most $d$,
    whose coefficients have absolute value bounded by $H$.
    Let $\alpha \neq \beta$ be distinct roots of $P$.
    Then 
    \[
        |\alpha - \beta| 
        > 
        \frac{\sqrt{3}}
             {(d+1)^{d + 2}H^{d - 1}}.
    \]
\end{theorem}

We obtain the following more explicit version of Lemma \ref{Lemma: bounding modulus of eigenvalues away from 1}:

\begin{lemma}\label{Lemma: bounding modulus of eigenvalues away from 1 (Appendix)}
    Let $\lambda$ be a complex algebraic number of degree $d$ and height $H$.
    Assume that $|\lambda| \neq 1$. 
    Then 
    \[ 
        \left||\lambda| - 1\right| 
        > 
        \frac{1}
            {\sqrt{3} (2d + 2)^{2d + 3} 2^{2d} (d!)^{2d} (d + 1)^{2d(d - 1)} H^{2d^2}}.
    \]
\end{lemma}
\begin{proof}
    The numbers $\lambda$ and $\bar{\lambda}$ are roots of the same minimal polynomial
    $P$ of degree $d$ and height $H$.
    It follows that the number $|\lambda|^2 = \lambda \bar{\lambda}$ is a root of the polynomial 
    $Q(x) = \operatorname{res}_z \left( P(x), x^d P(z/x) \right)$,
    where $\operatorname{res}_z(A,B)$ 
    denotes the resultant of the polynomials 
    $A,B \in \Q[x][z]$
    with coefficients in the integral domain 
    $\Q[x]$,
    cf.~\text{e.g.}~\cite[p. 159]{Cohen}.

    The degree of $Q(x)$ is at most $2d$.
    By \cite[Theorem 10]{ResultantBound} the height of $Q(x)$ is bounded by 
    \[
        H' = d! (d + 1)^{d - 1} H^d.
    \]
    The polynomial $Q(x)(x - 1)$ has degree at most $2d + 1$ and height at most 
    $2H'$.

    It follows from Theorem \ref{Theorem: Mignotte bound} that 
    \[
        \left||\lambda|^2 - 1\right|
        >
        \frac{\sqrt{3}}
             {(2d + 2)^{2d + 3}(2H')^{2d}}
        .
    \]
    Note that $\left||\lambda|^2 - 1\right| = \left||\lambda| - 1\right| \cdot \left||\lambda| + 1\right|$.
    If $|\lambda| > 2$ then the claim is trivial, so we may assume that $\left||\lambda| + 1\right| \leq 3$,
    yielding 
    \[
        \left||\lambda| - 1\right|
        >
        \frac{\sqrt{3}}
             {3 (2d + 2)^{2d + 3}(2H')^{2d}}
        =
        \frac{1}
            {\sqrt{3} (2d + 2)^{2d + 3} 2^{2d} (d!)^{2d} (d + 1)^{2d(d - 1)} H^{2d^2}}.
    \]

\end{proof}

%% file: app-Kronecker.tex
\section{Proof of Theorem \ref{theorem: General Quantitative Kronecker}}\label{Appendix: Kronecker}

Recall Dirichlet's theorem on simultaneous Diophantine approximation:

\begin{theorem}[Dirichlet]\label{Theorem: Dirichlet}
    Let $\varphi_1,\dots,\varphi_N \in \R$ be arbitrary real numbers.
    Let $M \in \R$ with $M \geq 1$.
    Then there exist integers 
    $q, p_1,\dots,p_N$ 
    with 
    $1 \leq q \leq M$
    such that 
    \[
        |q \varphi_j - p_j| < \tfrac{1}{q M^{1/N}}.
    \]
\end{theorem}

Throughout this section, let $\norm{\cdot}$ denote the distance to the closest integer.
We recall that Kronecker's theorem has two equivalent formulations: a discrete one and a continuous one.

\begin{theorem}[Kronecker's Theorem - Discrete Formulation]
	Let $\varphi_1,\dots,\varphi_N$ be real numbers, linearly independent over $\Q$.
	Let $\zeta_1,\dots,\zeta_N$ be arbitrary real numbers.
	Let $\varepsilon > 0$.
	Then there exists a real number $t$ such that for all $j$:
	\[
	\norm{\varphi_j t - \zeta_j} < \varepsilon.
	\]
\end{theorem}

\begin{theorem}[Kronecker's Theorem - Continuous Formulation]
	Let $1, \varphi_1,\dots,\varphi_N$ be real numbers, linearly independent over $\Q$.
	Let $\zeta_1,\dots,\zeta_N$ be arbitrary real numbers.
	Let $\varepsilon > 0$.
	Then there exists an integer $t$ such that for all $j$:
	\[
	\norm{\varphi_j t - \zeta_j} < \varepsilon.
	\]
\end{theorem}

The standard proof of equivalence of the two formulations in particular allows us to translate a quantitative version 
of the continuous formulation into a Quantitative version of the discrete formulation:

\begin{corollary}\label{Corollary: Kronecker integer bound from real bound}
	Let $\varphi_1,\dots,\varphi_N$ and $\zeta_1,\dots,\zeta_N$ be arbitrary real numbers.
	Let $\varepsilon > 0$.
	Let $q, p_1, \dots, p_N$ be integers such that 
	$|q \varphi_j - p_j| < \varepsilon$.
	If there exists a real number $0 \leq t \leq T$ such that for all $j$ we have
	\[
	\norm{(q\varphi_j - p_j)t - \zeta_j} < \varepsilon/2,
	\]
	then there exists an integer 
	$k \leq |q|T$ 
	such that for all $j$ we have
	\[
	\norm{\varphi_j k - \zeta_j} < \varepsilon.
	\]
\end{corollary}
\begin{proof}
	By assumption there exist integers $r_1, \dots, r_N$ such that we have 
	\[
	|(q\varphi_j - p_j)t - \zeta_j - r_j| < \varepsilon / 2.
	\]
	Write $t = \ell + \delta$ with $\ell \in \Z$ and $|\delta| \leq \tfrac{1}{2}$.
	We obtain:
	\[
	|q \ell \varphi_j  + q \delta \varphi_j - \ell p_j - \delta p_j - \zeta_j - r_j| < \varepsilon / 2.
	\]
	It follows that 
	\[
	\norm{q\ell \varphi_j - \zeta_j}
	\leq 
	|q \ell \varphi_j - \ell p_j - \zeta_j - r_j|
	\leq 
	|q \ell \varphi_j  + q \delta \varphi_j - \ell p_j - \delta p_j - \zeta_j - r_j|
	+
	|q \delta \varphi_j - \delta p_j|
	< \varepsilon.
	\]
	Thus, we may let $k = q \ell$.
\end{proof}

\begin{theorem}\label{Theorem: quantitative complex Kronecker}
    Let $\lambda_1,\dots,\lambda_N$ and $\alpha_1,\dots,\alpha_N$ be complex numbers of modulus $1$.
    Let $1/2 > \varepsilon > 0$ be a positive real number.
    Assume that $\lambda_1,\dots,\lambda_N$ are algebraic numbers such that the numbers 
    $
        \log \lambda_1,\dots,\log \lambda_N, 2\pi i
    $
    are linearly independent over $\Q$.
    Let $d$ be the degree of the field extension $\Q(\lambda_1,\dots,\lambda_N)$ over $\Q$ .
    Let $A_1,\dots,A_N \geq \exp(1)$ be upper bounds on the heights of $\lambda_1,\dots,\lambda_N$.
    Then there exists a positive integer 
    \[
        t \leq 
        8 \pi
        \left(\tfrac{2\pi}{\varepsilon}\right)^N
        \left(
            2N
            \left(\tfrac{2\pi}{\varepsilon}\right)^N
            \left\lceil
                \tfrac{4\pi}{\varepsilon} 
                \log \tfrac{4\pi N}{\varepsilon}
            \right\rceil
        \right)
        ^
        {
            \left(16(N + 1)d\right)^{2(N + 3)}\log A_1 \cdot \dots \cdot \log A_N
        }
    \] 
    such that 
    $
        |\lambda_j^t - \alpha_j| < \varepsilon
    $
    for all $j \in \{1,\dots,N\}$.
\end{theorem}
\begin{proof}
    Let $\log$ denote the determination of the logarithm where the imaginary part of $\log z$
    is in the interval $[0,2\pi)$.
    Write $\log \lambda_j = 2 \pi i \vartheta_j$
    and $\log \alpha_j = 2 \pi i \beta_j$.
    Let 
    $B = \left(\tfrac{2\pi}{\varepsilon}\right)^{N}$.
    Using Theorem \ref{Theorem: Dirichlet}, choose integers $q,p_1,\dots,p_N$ with $1 \leq q \leq B$ such that 
    \[
        |q \vartheta_j - p_j| < \tfrac{1}{B^{1/N}} = \tfrac{\varepsilon}{2\pi}.
    \]
    Let 
    \[
        M = \left\lceil \tfrac{4\pi}{\varepsilon} \log \tfrac{4\pi N}{\varepsilon}  \right\rceil.
    \]
    Let $m \in \Z^N\setminus\{0\}$ with $|m| \leq M$.
    Then 
    \begin{align*}
        &\left|
            m_1(q\vartheta_1 - p_1) + \dots + m_N(q \vartheta_N - p_N) 
        \right|\\
        &=
        \tfrac{1}{2\pi}
        \left|
            m_1(q 2\pi i \vartheta_1 - p_1 2\pi i) + \dots + m_N(q 2\pi i \vartheta_N - p_N 2\pi i)
        \right|\\
        &=
        \tfrac{1}{2\pi}
        \left|
            q m_1 2\pi i \vartheta_1 + \dots + q m_N 2\pi i \vartheta_N
            - (m_1p_1 + \dots + m_N p_N) 2\pi i
        \right|\\
        &=
        \tfrac{1}{2\pi}
        \left|
            q m_1 \log \lambda_1 + \dots + q m_N \log \lambda_N
            - 2(m_1p_1 + \dots + m_N p_N) \log (-1)
        \right|
        .
    \end{align*}
    By assumption the above quantity is non-zero, so Theorem \ref{Theorem: Baker} yields a uniform lower bound 
    \[
        \delta = \tfrac{1}{2\pi} \mathcal{B}^{-(16(N + 1)d)^{2(N + 3)}\log A_1 \cdot \dots \cdot \log A_N},
    \]
    where $\mathcal{B}$ is a bound on the size of the coefficients $qm_j$ and $2(m_1p_1 + \dots + m_N p_N)$.
    We have by construction $|q| \leq B$ and $|m_j| \leq M$.
    Since $\theta_j \leq 1$ we may choose $p_j \leq q \leq B$.
    It follows that we may choose
    $\mathcal{B} = 2NMB$.

    We have hence established an estimate
    \[
        \left|
            m_1(q\vartheta_1 - p_1) + \dots + m_N(q \vartheta_N - p_N) 
        \right|
        > \delta
        > 0
    \]
    for all $m \in \Z^N\setminus\{0\}$ with $|m| \leq M$. 
    Now, Theorem \ref{Theorem: Turan Kronecker} asserts the existence of a real number 
    $t_0 \in [0, 4/\delta]$ 
    and integers 
    $s_1,\dots,s_N \in \Z$ 
    such that 
    \[
        |(q\vartheta_j - p_j) t_0 - \beta_j - s_j| < \tfrac{\varepsilon}{4 \pi}.
    \] 
    Corollary \ref{Corollary: Kronecker integer bound from real bound} yields the existence of a positive integer
    $t \leq \tfrac{4}{\delta} \left(\tfrac{2\pi}{\varepsilon}\right)^{N}$
    and 
    $r_1,\dots,r_N \in \Z$ 
    such that 
    \[
        |\vartheta_j t - \beta_j - r_j| < \tfrac{\varepsilon}{2\pi}.
    \]
    By the mean value inequality it follows that 
    \[
        |\lambda^t - \alpha_j| < \varepsilon.
    \]
\end{proof}

\subsection{Admitting integer multiplicative relations}

\begin{proposition}\label{Proposition: integer multiplicative relations in normal form}
    Given complex algebraic numbers 
    $\lambda_1,\dots,\lambda_m$
    of modulus $1$ we can compute in polynomial space 
    positive integers 
    $1 \leq s \leq m$,
    $1 \leq j_1 \leq \dots \leq j_s \leq m$,
    $\ell \in \N$,
    and multi-indexes 
    $\varepsilon_{j} \in \Z^{s}$
    for 
    $j = 1,\dots,m$
    such that 
    $\lambda_{j_1},\dots,\lambda_{j_s}$
    do not admit any integer multiplicative relations
    and
    $
        \lambda_j^{\ell}
        =
        (\lambda_{j_1},\dots,\lambda_{j_s})^{\varepsilon_j}
    $
    for $j = 1,\dots,m$.
\end{proposition}
\begin{proof}
    By Theorem \ref{Theorem: Masser} we can compute in polynomial space a finite sequence of multi-indexes 
    $\beta_1,\dots,\beta_{m - s}$
    such that the free Abelian group 
    \[
        L = \Set{\alpha \in \Z^m}{(\lambda_1,\dots,\lambda_m)^\alpha = 1}
    \]
    is generated by $\beta_1,\dots,\beta_m$.
    Further, the size of the $\beta_j$'s is bounded polynomially in the sum of the heights and degrees of 
    $\lambda_1,\dots,\lambda_m$ and singly exponentially in $m$.

    Bring the matrix with rows 
    $\beta_1,\dots,\beta_{m - s}$ 
    into upper triangular form.
    This can be done in polynomial space.
    This yields indexes 
    $j_1,\dots,j_s$,
    positive numbers 
    $\ell_1,\dots,\ell_m$,
    and multi-indexes 
    $
    \eta_1, \dots, \eta_m
    $
    such that 
    \[
        \lambda_{j}^{\ell_j}
        =
        \left(
            \lambda_{j_1}
            ,\dots,
            \lambda_{j_s}
        \right)^{\eta_j}.
    \]
    Let 
    $\ell = \operatorname{lcm}(\ell_1,\dots,\ell_m)$
    and
    $\varepsilon_{j} = \ell/\ell_j \eta_j$. 
\end{proof}

Note that the bitsize of $\ell$ and $\varepsilon_1,\dots,\varepsilon_m$ are bounded polynomially in the input data,
but their total size may be exponential. 

\begin{proposition}\label{Proposition: closure of range}
Let $1 \leq s \leq m$ be positive integers.
Let $\ell \in \Z$.
Let $\varepsilon_{s + 1},\dots,\varepsilon_{m} \in \Z^{s}$ be multi-indexes.
Let 
\[
    f \colon \mathbb{T}^s \to \mathbb{T}^m,
    \;
    f(z_1,\dots,z_s)
    =
    \left(
        z_1^{\ell},\dots,z_s^{\ell},
        (z_1,\dots,z_s)^{\varepsilon_{s + 1}},
        \dots,
        (z_1,\dots,z_s)^{\varepsilon_{m}}
    \right).
\]
Then, with respect to the $\ell^{\infty}$-norm, $f$ is Lipschitz-continuous with Lipschitz constant
\[
    \max\left\{\ell, \sum_{k = 1}^s |\varepsilon_{s+1, k}|,\dots,\sum_{k = 1}^s |\varepsilon_{m, k}|\right\}.
\]
\end{proposition}
\begin{proof}
    Observe that $f$ extends to a differentiable function of type $\C^s \to \C^m$.
    Let $(z_1,\dots,z_s) \in \mathbb{D}^s$ be a point in the unit polydisk.
    Let $Df(z_1,\dots,z_s)$ denote the Jacobian of $f$ at $(z_1,\dots,z_s)$.
    By the mean value inequality it suffices to compute a bound on the operator norm of 
    $Df(z_1,\dots,z_s)$.
    An elementary calculation shows 
    \[
        \norm{Df(z_1,\dots,z_s)}_{\infty}
        =
        \max\left\{\ell, \sum_{k = 1}^s |\varepsilon_{s+1, k}|,\dots,\sum_{k = 1}^s |\varepsilon_{m, k}|\right\}.
    \]
\end{proof}

\begin{proposition}
    Let $1 \leq s \leq m$ be positive integers.
    Let $\ell \in \Z$.
    Let $\varepsilon_{s + 1},\dots,\varepsilon_{m} \in \Z^{s}$ be multi-indexes.
    Let 
    $\lambda_1,\dots,\lambda_s$ 
    be complex algebraic numbers of modulus $1$ which do not admit any integer multiplicative relations. 
    Let 
    \[
    f \colon \mathbb{T}^s \to \mathbb{T}^m,
    \;
    f(z_1,\dots,z_n)
    =
    \left(
        z_1^{\ell},\dots,z_s^{\ell},
        (z_1,\dots,z_s)^{\varepsilon_{s + 1}},
        \dots,
        (z_1,\dots,z_s)^{\varepsilon_{m}}
    \right).
    \]
    Then the closure of the sequence 
    $(f(\lambda_1^t,\dots,\lambda_s^t))_{t \in \N}$
    is equal to the range of $f$ over $\mathbb{T}^s$.
\end{proposition}
\begin{proof}
    By Kronecker's theorem \ref{Theorem: Kronecker} the sequence 
    $((\lambda_1^t,\dots,\lambda_s^t))_{t \in \N}$
    is dense in the torus $\mathbb{T}^s$.
    Let $A = \Set{(\lambda_1^t,\dots,\lambda_s^t)}{t \in \N}$.
    Since $f$ is continuous we have 
    $f(\clos{A}) \supseteq \clos{f(A)}$.
    Since $\mathbb{T}^s$ is compact,
    the range $f(\mathbb{T}^s)$ is closed,
    so that we have 
    $f(\mathbb{T}^s) \supseteq \clos{f(A)}$. 
    It follows that 
    $f(\clos{A}) = f(\mathbb{T}^s) = \clos{f(A)}$.
\end{proof}

\begin{theorem}\label{theorem: General Quantitative Kronecker (Appendix)}
    Let $(\lambda_1,\dots,\lambda_m)$ be complex algebraic numbers of modulus $1$.
    Assume that the numbers 
    $2\pi i, \log \lambda_1,\dots, \log \lambda_s$ 
    are linearly independent over the rationals,
    where $0 \leq s \leq m$.
    Let $d$ be the degree of the field extension
    $\Q(\lambda_1,\dots,\lambda_s)$.
    Let 
    $A_1,\dots,A_s \geq \exp(1)$
    be upper bounds on the heights of $\lambda_1,\dots,\lambda_s$. 
    Let 
    $\ell \in \N$,
    and 
    $\varepsilon_{s + 1},\dots,\varepsilon_m \in \Z^s$
    be such that 
    \[
        \lambda_j^{\ell} = (\lambda_1,\dots,\lambda_s)^{\varepsilon_j}
    \]
    for all $j = s + 1,\dots, m$.
    By convention, if $s = 0$ the right-hand side of the above equation is to 
    be taken equal to $1$.

    Let 
    \[
        L 
        =
        \max\left\{\ell, \sum_{k = 1}^s |\varepsilon_{s+1, k}|,\dots,\sum_{k = 1}^s |\varepsilon_{m, k}|\right\}.
    \]

    Let 
    $\alpha_1,\dots,\alpha_m \in \mathbb{T}^m$ 
    be such that any rational linear relation between the numbers 
    $2\pi i, \log \lambda_1,\dots,\log \lambda_m$
    is also satisfied by the numbers 
    $2\pi i, \log \alpha_1,\dots,\log \alpha_m$.
    Let $\varepsilon > 0$.
    Then there exists a positive integer 
    \[
        t 
        \leq 
            8\pi \ell 
            \left(
                \tfrac{2\pi L}{\varepsilon}
            \right)^s
            \left(
                2s
                \tfrac{2\pi L}{\varepsilon}
                \left\lceil
                    \tfrac{4\pi L}{\varepsilon}
                    \log \tfrac{4\pi s L}{\varepsilon}
                \right\rceil 
            \right)
            ^
            {
                \left(
                    16(s + 1)d
                \right)^{2(s + 3)}
                \log A_1 \cdot \dots \cdot  \log A_s
            }
        + \ell 
    \]
    such that 
    $
        \left|\lambda_j^t - \alpha_j\right| < \varepsilon
    $
    for $j = 1,\dots,m$.
\end{theorem}
\begin{proof}
Divide the sequence  
$(\lambda_1^t, \dots, \lambda_m^t)_{t \in \N}$
into $\ell$ disjoint subsequences
\[ 
    (\lambda_1^{\ell t + j},\dots,\lambda_m^{\ell t + j})_{t \in \N}
    =
    (
        \lambda_1^j \lambda_1^{\ell t}, \dots, \lambda_s^j \lambda_s^{\ell t},
        \lambda_{s + 1}^j (\lambda_1,\dots,\lambda_s)^{t\varepsilon_{s + 1}},
        \dots,
        \lambda_{m}^j (\lambda_1,\dots,\lambda_s)^{t\varepsilon_{m}}
    )_{t \in \N}.
\]
for 
$j = 0, \dots, \ell - 1$.

By Proposition \ref{Proposition: closure of range}, the closure of the sequence 
$(\lambda_1^{\ell t + j},\dots,\lambda_m^{\ell t + j})_t$
is the set 
\[
    C_j 
    = 
    \Set{(\lambda_1^j z_1^e,\dots,\lambda_s^j z_s^e,
          \lambda_{s + 1}^j (z_1,\dots,z_s)^{\varepsilon_{s + 1}},
          \dots,
          \lambda_{m}^j (z_1,\dots,z_s)^{t\varepsilon_{m}}
         )}
        {
            (z_1,\dots,z_s) \in \mathbb{T}^s
        }.
\]
Hence, the closure of the sequence 
$(\lambda_1^t,\dots,\lambda_m^t)_t$
is the union of the sets $C_j$.

Let 
$(z_1,\dots,z_m)$ 
be contained in the closure of the sequence 
$(\lambda_1^t,\dots,\lambda_m^t)$.
Let $j$ be such that $(z_1,\dots,z_m) \in C_j$.

Since the numbers $\lambda_j$ have modulus $1$, the function
\[ 
    f_j(z_1, \dots, z_s) 
    =
    \left(
            \lambda_1^j z_1^{\ell},\dots, \lambda_s^j z_s^{\ell},
            \lambda_{s + 1}^j (z_1,\dots,z_s)^{\varepsilon_{s + 1}},
            \dots,
            \lambda_m^j (z_1,\dots,z_s)^{\varepsilon_{m}}
    \right)
\]
is Lipschitz-continuous with Lipschitz constant $L$.
Let 
$t(\varepsilon)$ 
denote the bound from
Theorem \ref{Theorem: quantitative complex Kronecker},
as a function of $\varepsilon$.
By definition there exists $t \leq t(\varepsilon/L)$ such that 
\[ 
    |f_j(\lambda_1^t,\dots,\lambda_s^t)
    -
    (z_1,\dots,z_m)|
    <
    \varepsilon.
\]
We have 
$
    f_j(\lambda_1^t,\dots,\lambda_s^t)
    = 
    (
        \lambda_1^{\ell t + j},\dots,\lambda_m^{\ell t + j}
    )
$.
The sequence index $\ell t + j$ is smaller than 
$\ell\left(t(\varepsilon/L) + 1\right)$.
\end{proof}

%% file: app-Separation.tex
\section{Proof of Lemma \ref{Lemma: compact semialgebraic set separation bound}}\label{Appendix: separation}

\begin{lemma}\label{Lemma: compact semialgebraic set separation bound (Appendix)}
    There exists an integer function 
    $
        \operatorname{Sep}(n, d, \tau) \in 2^{\left(\tau d\right)^{n^{O(1)}}}
    $
    with the following property:

    Let $K$ and $L$ be compact semialgebraic sets of complexity at most $(n,d,\tau)$.
    Assume that every $x \in K$ has positive euclidean distance to $L$.
    Then $\inf_{x \in K} \operatorname{dist}_{\ell^2}(x,L) > 1/\operatorname{Sep}(n,d,\tau)$.
\end{lemma}
\begin{proof}
    If either $K$ or $L$ are empty then the result is trivial.
    Thus, let us assume that both sets are non-empty.

    Consider the semialgebraic set 
    \[
        S = \Set{(x,y) \in K \times L}
            {\forall z \in L. \left(\norm{x - z}_2^2 \geq \norm{x - y}\right)}.
    \]
    By Theorem \ref{Theorem: singly exponential quantifier elimination precise} the set $S$ has complexity 
    $(2n, \left(d \tau\right)^{n^{O(1)}}, \left(d,\tau\right)^{n^{O(1)}})$.
    By compactness, the distance $\operatorname{dist}_{\ell^2}(x, L)$ is attained in a point $y \in L$ 
    for all $x \in K$, so that for all $x \in K$ there exists $y \in L$ such that $(x,y) \in S$.

    We clearly have 
    \begin{equation}\label{eq: compact separation bound eq 1}
        \inf_{x \in K} \operatorname{dist}_{\ell^2}(x,L) = \inf_{(x,y) \in S} \norm{x - y}_2^2.
    \end{equation}
    The right-hand side of \eqref{eq: compact separation bound eq 1} is a polynomial,
    so that the result follows from Theorem \ref{Theorem: Jeronimo}.
\end{proof}

%% file: app-Recurrent.tex
\section{Proof of Lemma \ref{Lemma: recurrent escape bound}}
\label{Appendix: Recurrent}

We will first prove the following weaker version of Lemma \ref{Lemma: recurrent escape bound},
where we only establish an escape bound and a lower bound on the distance to $K$ for initial points 
in $K_{\operatorname{rec}}$.

\begin{lemma}\label{Lemma: recurrent escape bound for points in K}
    There exists an integer function 
    $
        \operatorname{Rec}_0(n, d, \tau) \in 2^{\left(\tau d \right)^{n^{O(1)}}}
    $
    with the following property:

    Let $A \in \mathbb{A}^{n \times n}$ be a matrix in real Jordan normal form.
    Assume that the minimal polynomial of $A$ has rational coefficients whose bitsize is bounded by $\tau$.
    Let $K \subseteq \R^n$ be a semialgebraic set of complexity at most $(n, d,\tau)$.
    If every point $x \in K_{\operatorname{rec}}$ escapes $K$ under iterations of $A$ then 
    for all $x \in K_{\operatorname{rec}}$ there exists 
    $t \leq \operatorname{Rec}_0(n, d,\tau)$
    such that 
    \[
        \operatorname{dist}_{\ell^2}(A^t x, K) > \frac{\sqrt{n}}{\operatorname{Rec}_0(n,d,\tau)}.
    \]
\end{lemma}

For $x \in V_{\operatorname{rec}}$, let 
\[
    \O_A(x) = \Set{A^t x}{t \in \N}
\]
denote the orbit of $x$ under $A$. 
Let $\clos{\O_A(x)}$ denote its closure.
By Kronecker's theorem \ref{Theorem: Kronecker}, the set $\clos{\O_A(x)}$ is semialgebraic.

The sequence $(A^t x)_{t \in \N}$ is dense in $\clos{\O_A(x)}$ by definition.
A combination of Theorem \ref{theorem: General Quantitative Kronecker} and Theorem \ref{Theorem: Basu-Roy radius bound}
yields a quantitative refinement of this qualitative statement:

\begin{lemma}\label{Lemma: modulus of density}
    There exists an integer function  
    $
        \operatorname{D}(n, d, \tau, P) \in 2^{\left(\tau P d\right)^{n^{O(1)}}}
    $
    with the following property: 

    Let $A$ be a matrix in real Jordan normal form.
    Assume that the characteristic polynomial of $A$ has rational coefficients whose bitsize is bounded by $\tau$.
    Let $K$ be a compact semialgebraic set of complexity at most $(n, d, \tau)$.
    Let $P$ be a positive integer.
    Then for all $x \in K_{\operatorname{rec}}$ and all $y \in \overline{\O_A(x)}$ 
    there exists $t \leq \operatorname{D}(n,d,\tau,P)$
    such that 
    \[ 
        \norm{A^t x - y}_2 < 2^{-P}.
    \]
\end{lemma}
\begin{proof}
	Let 
	$\operatorname{Kron}(n,\tau,P) \in  2^{(\tau P)^{n^{O(1)}}}$
	be the function from Corollary \ref{Corollary: qualitative Kronecker bound}.
    Let $\operatorname{Bound}(n,d,\tau) = 2^{\tau d^{\beta (n + 1)}}$, 
    where $\beta$ is the constant from Theorem \ref{Theorem: Basu-Roy radius bound}. 

   	Put
    \[
        \operatorname{D}(n, d, \tau, P) 
        =
        \operatorname{Kron}\left(n,\tau, P + \left\lceil \log(n) + \log\left(\operatorname{Bound}(n,d,\tau)\right) \right\rceil\right).
    \]
	It is easy to see that 
	$\operatorname{D}(n, d, \tau, P) \in 2^{\left(\tau P d\right)^{n^{O(1)}}}$
    as claimed.
	
	To prove that $\operatorname{D}$ has the desired properties, let $A$ and $K$ be a matrix and a compact semialgebraic set as above.
	Let $P$ be a positive integer.
	For a matrix $Q = (q_{i,j})_{i,j = 1}^n \in \R^{n \times n}$
	Let 
	\[
    \norm{Q}_F
    =
    \left(\sum_{i = 1}^n \sum_{j = 1}^n q_{i,j}^2 \right)^{1/2}
	\]
	denote the Frobenius norm of $Q$.
	The Frobenius norm is sub-multiplicative and hence satisfies
	\[
		\norm{Q \cdot x}_2
		\leq 
		\norm{Q}_F \cdot \norm{(x,\dots,x)}_F
		=
		\norm{Q}_F \cdot \sqrt{n} \cdot \norm{x}_2
	\]
	for all $ x\in \R^n$.

	Let $C = \operatorname{Bound}(n,d,\tau)$.
	Let $x \in K_{\operatorname{rec}}$.
	Let $y \in \clos{\O_A(x)}$.
	Then, by Corollary \ref{Corollary: qualitative Kronecker bound}, there exist
	$t \leq \operatorname{Kron}(n, \tau, P + \left\lceil \log(n) + \log\left(\operatorname{Bound}(n,d,\tau)\right) \right\rceil)$
	and 
	$Q \in \mathbb{R}^{n \times n}$
	such that $y = Q x$ and 
	\[
		\norm{A^t - Q}_F < 2^{-P - \left\lceil \log(n) + \log (C) \right\rceil} \leq \frac{2^{-P}}{\sqrt{n}C}
	\]
	Hence: 
	\[
		\norm{A^t x - y}_2 = \norm{(A^t - Q)x}_2 \leq \norm{A^t - Q}_F \cdot \sqrt{n} \cdot |x| < 2^{-P}.
	\]
	The last inequality uses that by Theorem \ref{Theorem: Basu-Roy radius bound} we have $|x| \leq C$.
\end{proof}

Let $A \in \mathbb{A}^{n \times n}$ be a matrix in real Jordan normal form. 
Assume that the characteristic polynomial of $A$ has rational coefficients of bitsize at most $\tau$. 
Let $K \subseteq \R^n$ be a compact semialgebraic set of complexity at most $(n,d,\tau)$.
Assume that every point $x \in K_{\operatorname{rec}}$ escapes $K$ under iterations of $A$.

By definition, a point $x \in K_{\operatorname{rec}}$ escapes $K$ under iterations of $A$ if and only if 
there exists $y \in \clos{\O_A(x)}$ such that $\operatorname{dist}_{\ell^2}(y, K) > 0$.
Our next goal is to sharpen this to a uniform lower bound on 
$\inf_{x \in K} \sup_{y \in \clos{\O_A(x)}} \operatorname{dist}_{\ell^2}(y,K)$.
The main idea is to employ Theorem \ref{Theorem: Jeronimo}.
Since the function $g(x) = \sup_{y \in \clos{\O_A(x)}} \operatorname{dist}_{\ell^2}(y, K)$ is not a polynomial, we construct an auxiliary compact semialgebraic set in higher dimension that allows us to reduce the problem of finding a uniform lower bound on $g(x)$ to the problem of finding a uniform lower bound on a polynomial.
The idea is essentially the same as that of the proof of Lemma \ref{Lemma: compact semialgebraic set separation bound}.

Let 
\[
	\clos{\O_A} 
	= 
	\Set{(x,y) \in \R^{2n}}
		{x \in K_{\operatorname{rec}}, y \in \clos{\O_A(x)}}.
\]
Let 
\[
	S = 
		\Set{(x,y,z) \in \R^{3n}}
			{x \in K_{\operatorname{rec}},  y \in \clos{\O_A(x)}, z \in K}
	= \clos{\O_A} \times K.
\]
By compactness of $K$, the number 
\begin{equation}\label{eq: protrusion}
	\eta = \min_{x \in K_{\operatorname{rec}}} \max_{y \in \clos{\O_A(x)}} \min_{z \in K} \norm{y - z}^2_2
\end{equation}
is strictly positive.
Letting $\operatorname{D}(n, d, \tau, P)$ denote the function from Lemma \ref{Lemma: modulus of density},
observe that every point $x \in K_{\operatorname{rec}}$ escapes $K$ under at most 
\[
    \operatorname{D}\left(n, d, \tau, \left\lceil\log\left(1/\eta\right)\right\rceil\right)
\]
iterations of $A$.
To obtain a bound on the escape time it hence suffices to obtain a bound of \eqref{eq: protrusion} away from zero.
This is achieved by expressing \eqref{eq: protrusion} as the minimum of a polynomial over a compact semialgebraic set.
To this end, we consider the set
\[
	S' = \Set{(x,y,z) \in \R^{3n}}
	{x \in K_{\operatorname{rec}},  y \in \clos{\O_A(x)},  z \in K,
	\forall w \in \clos{\O_A(x)}. \operatorname{dist}_{\ell^2}(y,K) \geq \operatorname{dist}_{\ell^2}(w,K)}.
\]
Observe that 
\begin{equation}\label{eq: expressing protrusion as minimum of a polynomial}
	\min_{x \in K_{\operatorname{rec}}} \max_{y \in \clos{\O_A(x)}} \min_{z \in K} 
		\norm{y - z}^2_2
	=
	\min_{(x,y,z) \in S'} \norm{y - z}^2_2.
\end{equation}
This leads to two problems: 
\begin{enumerate}
	\item Bound the complexity of $\clos{\O_A}$ (and thus of $S$) in terms of the complexity of $A$ and $K$.
	\item Bound the complexity of $S'$ in terms of the complexity of $S$.
\end{enumerate}

Let us first bound the complexity of $\clos{\O_A}$ in terms of that of $A$ and $K$.

Up to suitably permuting the dimensions, which does not affect the complexity, we may assume that $A$ takes the following form: 
\begin{equation}\label{eq: orbit complexity matrix normal form}
    A = 
        \begin{pmatrix}
            I_{m_1 \times m_1} &         &           &        &           &         \\ 
                 & -I_{m_2 \times m_2}   &           &        &           &         \\ 
                 &         & \Lambda_1 &        &           &         \\
                 &         &           & \ddots &           &         \\ 
                 &         &           &        & \Lambda_{m_3} &         \\ 
                 &         &           &        &           & B  
        \end{pmatrix}
\end{equation}
where $I_{m \times m}$ is the $m\times m$ identity matrix, $\Lambda_1,\dots,\Lambda_{m_3}$ are $2\times 2$ matrices corresponding to the genuinely complex eigenvalues of $A$ of modulus $1$, and $B$ is some matrix.
With this convention, $K_{\operatorname{rec}}$ is the set of all $x \in K$ with $x_j = 0$ for $j > m_1 + m_2 + m_3$.

Let 
\[
    A_0 = 
    \begin{pmatrix}
        I_{m_1 \times m_1} &         &           &        &           &         \\ 
             & -I_{m_2 \times m_2}   &           &        &           &         \\ 
             &         & \Lambda_1 &        &           &         \\
             &         &           & \ddots &           &         \\ 
             &         &           &        & \Lambda_{m_3} &         \\ 
             &         &           &        &           & \textbf{0}  
    \end{pmatrix}
\]
where $\textbf{0}$ is the zero matrix of the same dimensions as $B$ in \eqref{eq: orbit complexity matrix normal form}.
Clearly, the orbits $\O_A(x)$ and $\O_{A_0}(x)$ coincide for all $x \in K_{\operatorname{rec}}$.

Let 
\[ 
    L = \Set{(\alpha_1,\dots,\alpha_{m_3}) \in \Z^{m_3}}
            {\Lambda_1^{\alpha_1} \cdot \dots \cdot \Lambda_{m_3}^{\alpha_{m_3}} = I_{2 \times 2}}.
\]
By Theorem \ref{Theorem: Masser} the abelian group $L$ has a basis 
$\beta_1,\dots,\beta_\ell$ 
with the magnitudes of $\beta_1,\dots,\beta_\ell$ bounded polynomially in the data that is used to describe the algebraic entries of the matrices 
$\Lambda_1,\dots,\Lambda_{m_3}$ and singly exponentially in $n$.  
Thus, by our assumption on the characteristic polynomial of $A$, we have $\sum_{i,j} |\beta_{i,j}| = \left(n\tau\right)^{n^{O(1)}}$.

By Kronecker's theorem (Theorem \ref{Theorem: Kronecker}), the closure of the set $(A_0^t)_{t \in \N}$ in $\R^{n \times n}$ is given by the set of all matrices of the form 
\[
    \begin{pmatrix}
        I_{m_1 \times m_1} &         &           &        &           &         \\ 
             & \sigma \cdot I_{m_2 \times m_2}   &           &        &           &         \\ 
             &         & Z_1 &        &           &         \\
             &         &           & \ddots &           &         \\ 
             &         &           &        & Z_{m_3} &         \\ 
             &         &           &        &           & \textbf{0}  
    \end{pmatrix},
\]
where $\sigma \in \{-1,1\}$ and $Z_1, \dots, Z_{m_3}$ are $2\times 2$ matrices satisfying 
\begin{equation}\label{eq: matrix multiplicative relations}
    Z_1^{\beta_{i,1}} \cdot \dots \cdot Z_{m_3}^{\beta_{i,m_3}} = I_{2 \times 2}
\end{equation}
for all $i = 1, \dots, \ell$.

The size of the polynomials required to describe the relation \eqref{eq: matrix multiplicative relations} can be bounded by the following straightforward lemma:

\begin{lemma}
    Let $M_1,\dots,M_N$ be $2\times 2$ matrices with entries in a polynomial ring 
    $\Z[x_1,\dots,x_k]$.
    Let $d$ be a bound on the total degree of all matrix entries. 
    Let $\tau$ be a bound on the bitsize of the coefficients of all matrix entries.
    Then the entries of the product 
    $M_1 \cdot \dots \cdot M_N$
    have total degree at most 
    $N d$
    and coefficients bounded in bitsize by 
    $N \tau + k (N - 1) \log\left(d + 1\right)$.
\end{lemma}
% \begin{proof}
% 	\EN{TODO: proof added for completeness -- will probably remove to save space}
%     By induction on $N$.
%     For $N = 1$ there is nothing to prove.
%     Assume the claim holds true for $N \in \N$.
%     Consider the product 
%     $M_1 \cdot \dots \cdot M_N \cdot M_{N + 1}$.

%     Write $M = M_1 \cdot \dots \cdot M_N$.
%     By induction hypothesis the entries of the matrix 
%     $M$ 
%     have total degree at most $N d$
%     and coefficients bounded in bitsize by $N \tau + k (N - 1) \log\left(d + 1\right)$.

%     All entries of $M \cdot M_{N + 1}$
%     are of the form 
%     $PQ + RS$ 
%     with $P$ and $R$ entries of $M$ 
%     and $Q$ and $S$ entries of $M_{N + 1}$.
%     Using standard estimates on the bitsize of sum and product of integers 
%     we obtain that this polynomial has degree at most 
%     $N d + d = (N + 1)d$
%     and coefficients bounded in bitsize by 
%     $(N + 1) \tau + k N \log (d + 1)$.
% \end{proof}

Thus, the closure $\mathcal{M}$ of $(A_0^t)_{t \in \N}$ is an algebraic set that can be described using 
at most $n^2$ polynomials in $n^2$ variables, 
each of which has its degree bounded by $(n\tau)^{n^{O(1)}}$,
and coefficients bounded in bitsize by a polynomial in $n$ and $\tau$.

We can hence describe $\clos{\O_A}$ by the following formula: 
\begin{align*}
    &\exists a_{1,1},\dots,a_{1,n},\dots,a_{n,1},\dots,a_{n,n}.\\
    &\left(
        (x_1,\dots,x_n) \in K
        \land 
        \left(\bigwedge_{i = m_1 + m_2 + m_3 + 1}^n \left(x_i = 0\right)\right)
        \land 
        (a_{i,j})_{i,j = 1,\dots,n} \in \mathcal{M}
        \land 
        \left(\bigwedge_{i = 1}^n \left(y_i = \sum_{j = 1}^n a_{i,j} x_j\right)\right)
    \right).
\end{align*}

This formula involves polynomials in $n^2 + n$ variables.
Their degrees are bounded by an expression in
$d\tau + (n\tau)^{n^{O(1)}}$ 
and their coefficients are bounded in bitsize by an expression in
$(n \tau)^{O(1)}$.
By applying singly-exponential quantifier elimination (Theorem \ref{Theorem: singly exponential quantifier elimination precise}) we obtain that 
$\overline{\O_A}$
can be defined by a quantifier free formula involving 
polynomials of degree at most 
$(d \tau)^{n^{O(1)}}$
whose coefficients have bitsize at most 
$(d \tau)^{n^{O(1)}}$.
The same bounds hold true for the complexity of the set $S = \overline{\O_A} \times K$. 

Let us now bound the complexity of the set $S'$.
	For $w, y \in \clos{\O_A(x)}$ the predicate $\operatorname{dist}_{\ell^2}(y, K) \geq \operatorname{dist}_{\ell^2}(w, K)$ is expressed by the following first-order formula:
	\[
		\exists m,n,v.
		\left(
			v \in K \land n \in K \land v \in K
			\land 
			\left(
				d_{\ell^2}(y,v) < d_{\ell^2}(y,m) \lor d_{\ell^2}(w,v) < d_{\ell^2}(w, n)
				\lor 
				d_{\ell^2}(y, m) \geq d_{\ell^2}(w, n)
			\right) 
		\right).
	\]
	Let us write this as 
	$\exists m,n,v. \Phi(y,w,m,n,v)$.
	Hence, the predicate $\forall w \in \clos{\O_A(x)}. \operatorname{dist}_{\ell^2}(y,K) \geq \operatorname{dist}_{\ell^2}(w, K)$
	can be written as 
	\[
		\Psi(x, y) =
		\forall w. \exists m,n,v. \left(
					w \notin \clos{\O_A(x)} \lor \Phi(y,w,m,n,v) 
				\right).
	\]
	Thus,
	\[
		S' = \Set{(x,y,z) \in \R^{3n}}{(x,y,z) \in S \land \Psi(x,y)}.
	\]
	The formula 
    $\Psi(x,y)$ 
    involves polynomials of degree at most 
    $(d \tau)^{n^{O(1)}}$
    and coefficients bounded in bitsize by 
    $(d \tau)^{n^{O(1)}}$.
    Now, Theorem \ref{Theorem: singly exponential quantifier elimination precise} yields 
    a complexity bound for $S'$ of at most 
    $(3n, (d\tau)^{n^{O(1)}}, (d\tau)^{n^{O(1)}})$.

By Theorem \ref{Theorem: Jeronimo} we obtain the existence of a function 
\[
    \operatorname{LowerBound}'(n, d, \tau) \in 2^{\left(d \tau\right)^{n^{O(1)}}}
\]
such that the minimum of the polynomial \eqref{eq: expressing protrusion as minimum of a polynomial} over the set $S'$ is bounded from below by 
$\frac{2\sqrt{n}}{\operatorname{LowerBound}'(n, d,\tau)}$.
This means that for all $x \in K_{\operatorname{rec}}$ there exists $y \in \clos{\O_A(x)}$ such that 
\[
    \operatorname{dist}_{\ell^2}(y, K) > \frac{2\sqrt{n}}{\operatorname{LowerBound}'(n, d,\tau)}.
\]

Now, consider the function  
\[
    \operatorname{Rec}_0(n,d,\tau) 
    =
    \max
    \{
    \operatorname{LowerBound}'(n,d,\tau),
    \operatorname{D}(n, d, \tau, \left\lceil\log \left(\operatorname{LowerBound}'(n, d, \tau)\right)\right\rceil)
    \}
    \in 
    2^{(d \tau)^{n^{O(1)}}}.
\]

We claim that the function $\operatorname{Rec}_0$ has the property stated in Lemma \ref{Lemma: recurrent escape bound for points in K}.
Let $x \in K_{\operatorname{rec}}$. 
Let 
$y \in \clos{\O_A}(x)$ 
with 
$\operatorname{dist}_{\ell^2}(y, K) > \frac{2\sqrt{n}}{\operatorname{LowerBound}'(n, d,\tau)}$. 
Then there exists $t \leq \operatorname{Rec}_0(n,d,\tau)$ such that 
\[
    \norm{A^t x - y}_2 < \frac{1}{\operatorname{LowerBound}'(n,d,\tau)} \leq \frac{\sqrt{n}}{\operatorname{LowerBound}'(n,d,\tau)}.
\]
For all $z \in K$ we then have  
\[
    \norm{A^t x - z}_2 
        > 
        \norm{y - z}_2
        -
        \norm{A^t x - y}_2
        >
        = 
        \frac{\sqrt{n}}{\operatorname{LowerBound}'(n, d,\tau)}
        > 
        \frac{\sqrt{n}}{\operatorname{Rec}(n, d,\tau)}.
\] 

This concludes the proof of Lemma \ref{Lemma: recurrent escape bound for points in K}.
It remains to extend the result to all initial points $x \in V_{\operatorname{rec}}$.
We first treat the special case where $K_{\operatorname{rec}}$ is empty.

\begin{lemma}\label{Lemma: recurrent bound for empty K-rec}
    There exists an integer function 
    $
        \operatorname{EmptyBound}(n,d,\tau) \in 2^{\left(d \tau\right)^{n^{O(1)}}}
    $
    with the following property:

    Assume that $K_{\operatorname{rec}}$ is empty.
    Let $x \in V_{\operatorname{rec}}$.
    Then $\operatorname{dist}_{\ell^2}(x, K) > 1/\operatorname{EmptyBound}(n,d,\tau)$.
\end{lemma}
\begin{proof}
    The function $\operatorname{dist}_{\ell^2}(\cdot, V_{\operatorname{rec}})$ is linear, and thus in particular a polynomial.
    By assumption, 
    \[ 
        \inf_{x \in K} \operatorname{dist}_{\ell^2}(x, V_{\operatorname{rec}}) > 0.
    \]
    Now, Theorem \ref{Theorem: Jeronimo} yields a lower bound of the desired shape.
\end{proof}

We begin with a technical lemma, which combines quantifier elimination with Lemma 
\ref{Lemma: compact semialgebraic set separation bound}.

\begin{lemma}\label{Lemma: recurrent bound technical lemma}
    Assume that $K_{\operatorname{rec}}$ is non-empty.
    Let $C$ be a positive integer.
    Further assume that $K_{\operatorname{rec}}$ is contained in a ball of radius 
    $2^{{\left(d \tau\right)^{n^C}}}$.
    Let $x \in V_{\operatorname{rec}}$.
    If 
    $\operatorname{dist}_{\ell^2}(x, K_{\operatorname{rec}}) > 2^{-{\left(d \tau\right)^{n^C}}}$
    then 
    $
        \operatorname{dist}_{\ell^2}(x, K) > 2^{-(d \tau)^{n^{C + O(1)}}}
    $.
\end{lemma}
\begin{proof}
    Consider the function 
    $
        \operatorname{dist}_{\ell^2}(x, K)
    $
    on the set 
    \[
        L = \Set{x \in V_{\operatorname{rec}}}
                 {\norm{x}_2 \leq 2^{{\left(d \tau\right)^{n^C}}} \land \operatorname{dist}_{\ell^2}(x, K_{\operatorname{rec}}) \geq 2^{{-\left(d \tau\right)^{n^C}}}}.
    \]
    This set can be defined by the following first-order formula: 
    \begin{align*}
        &\forall (y_1,\dots,y_n) \in \R^n.
        \forall (b_0, \dots, b_m) \in \R^{n^C}.\\
        &\left(
            \left(
                (y_1,\dots,y_n) \in K_{\operatorname{rec}}
                \land 
                b_0 = 2
                \land 
                b_{i + 1} = b_i^{\left(d \tau\right)}
            \right)
            \rightarrow
            \norm{x}_2^2 \leq b_m^2
            \land 
            \norm{x - y}_2^2 \geq (1/b_m)^2
        \right)
    \end{align*}
    Applying quantifier elimination (Theorem \ref{Theorem: singly exponential quantifier elimination precise}), we find that the set $L$ can be be defined by a quantifier-free formula 
    involving polynomials of degree 
    $(d\tau)^{O(n^C)}$
    whose coefficients are bounded in bitsize by 
    $(d\tau)^{O(n^C)}$.

    It follows from Lemma \ref{Lemma: compact semialgebraic set separation bound} that every point in $L$ has distance from $K$ at least 
    \[ 
        1/\operatorname{Sep}(n, (d\tau)^{O(n^C)}, (d\tau)^{O(n^C)})
        =
        2^{\left(-(d \tau)^{O(n^C)}\right)^{n^{O(1)}}}
        =
        2^{-(d \tau)^{n^{C + O(1)}}}
    \]
\end{proof}

Let us now turn to the proof of Lemma \ref{Lemma: compact semialgebraic set separation bound}.

The function $\operatorname{Rec}(n,d,\tau)$ is majorised by $2^{(d\tau)^{n^R}}$ for some constant $R$.
Let $C \geq R + 1$ be such that $K$ is contained in a ball of radius $2^{(d\tau)^{n^C}}$.
Such a constant exists by Theorem \ref{Theorem: Basu-Roy radius bound}.
Let $D$ be a constant such that all $x \in V_{\operatorname{rec}}$ 
with $\operatorname{dist}_{\ell^2}(x, K_{\operatorname{rec}}) > 2^{-(d\tau)^{n^{C}}}$
satisfy $\operatorname{dist}_{\ell^2}(x, K) > 2^{-(d\tau)^{n^{C + D}}}$.
The constant $D$ exists thanks to Lemma \ref{Lemma: recurrent bound technical lemma}.

Let 
\[ 
    \operatorname{Rec}(n,d,\tau) = 
        \max\left\{
            \operatorname{EmptyBound}(n,d,\tau),
            \left\lceil\sqrt{n}2^{(d\tau)^{n^{C + D}}}\right\rceil,
            \left\lceil\frac{\sqrt{n}\operatorname{Rec}_0(n,d,\tau)}
                    {\sqrt{n} - 2^{-(d\tau)^{n^C}} \operatorname{Rec}_0(n,d,\tau)}
            \right\rceil
            \right\}
\]
Then, since we have chosen $C \geq R + 1$, we have $\operatorname{Rec}(n,d,\tau) \in 2^{(d\tau)^{n^{O(1)}}}$.

Let us now show that $\operatorname{Rec}$ has the desired property.
If $K_{\operatorname{rec}}$ is empty, then $\operatorname{Rec}$ has the desired property by construction.
Let us hence assume for the rest of the proof that $K_{\operatorname{rec}}$ is non-empty.

Let $x \in V_{\operatorname{rec}}$. 
If $\operatorname{dist}_{\ell^2}(x,K) > 2^{-(d\tau)^{C + D}}$ 
then $\operatorname{dist}_{\ell^2}(x, K) > \sqrt{n}/\operatorname{Rec}(n,d,\tau)$

Assume that $\operatorname{dist}_{\ell^2}(x,K) \leq 2^{-(d\tau)^{C + D}}$.
Then by Lemma \ref{Lemma: recurrent bound technical lemma} we have
$\operatorname{dist}_{\ell^2}(x, K_{\operatorname{rec}}) \leq 2^{-(d\tau)^{n^C}}$.
Choose $y \in K_{\operatorname{rec}}$ such that $\norm{x - y}_2 < 2^{-(d\tau)^{n^C}}$.
Then, by Lemma \ref{Lemma: recurrent escape bound for points in K} there exists 
$t \leq \operatorname{Rec}_0(n,d,\tau)$ 
such that 
\[ 
    \operatorname{dist}_{\ell^2}(A^t y, K) > \sqrt{n}/\operatorname{Rec}_0(n,d,\tau).
\]
Since $A$ is an isometry on $V_{\operatorname{rec}}$, we have 
\[
    \norm{A^t x - A^t y}_2 = \norm{x - y}_2 \leq 2^{-(d\tau)^{n^C}}.
\]
It follows that 
\[
    \operatorname{dist}_{\ell^2}(A^t x, K) 
    > 
    \sqrt{n}/\operatorname{Rec}_0(n,d,\tau) - 2^{-(d\tau)^{n^C}}
    >
    \sqrt{n}/\operatorname{Rec}(n,d,\tau).
\]

%% file: app-Nonrecurrent.tex
\section{Proof of Lemma \ref{Lemma: non-recurrent overall bound}}
\label{Appendix: non-recurrent}

Let $J_k$ be a real Jordan block of multiplicity $k$ corresponding to either a real eigenvalue $\Lambda$ or a complex pair $\Lambda=a\pm ib$. We use $t$ to denote positive integer time-steps.
\begin{equation*}
	J_{k}^{t}=\begin{bmatrix}
		\Lambda^{t}	&&	t\Lambda^{t-1}	&&	{t\choose 2}\Lambda^{t-1}	&&
		\cdots		&&	{t\choose k-1}\Lambda^{t-k+1}				\\
		0			&&	\Lambda^{t}		&&	t\Lambda^{t-1}				&&
		\cdots		&&	{t\choose k-2}\Lambda^{t-k+2}				\\
		\vdots	&&	\vdots	&&	\vdots	&&	\ddots	&&	\vdots			\\
		0		&&	0		&&	0		&&	\cdots	&&	t\Lambda^{t-1}	\\
		0		&&	0		&&	0		&&	\cdots	&&	\Lambda^{t}		\\
	\end{bmatrix}.
\end{equation*}
where $\Lambda$ can be considered  $\begin{bmatrix}a&-b\\b&a\end{bmatrix}$ or a scalar quantity depending on the type of eigenvalue. 

Thus we see that $x(t): = J_k^t x$ can be written component-wise as 
\[x_j (t) = \sum_{i=j}^{k} \binom{t}{i-j}\Lambda^{t-(i-j)}x_{i},\]
where the `components' $x_1 \dots x_k$ refer to scalars (if the eigenvalue is real) or vectors $(x_j^{(r)}, x_j^{(i)})$ if the eigenvalue is complex.

We will use absolute value signs to denote the Jordan norm $\Jnorm{\cdot}$ of a component $x_j$. Note this is the same as the $\ell_2$ norm since we consider only one component.

Then we have the following three lemmas, where we case split on the modulus of the eigenvalue $\Lambda$, which we notate as $\modul$, to obtain block-wise bounds on escape times for non-recurrent eigenspaces:
%\EL{Modulus replaced by $\modul$ and using a macro if one wants to change it easily.}

These lemmas have the same structure as those in \cite{CLOW20}, however for convenience we present the discrete case in full here, as \cite{CLOW20} formulates them in the continuous case.

\begin{lemma}[Polynomially expanding Jordan block - $\modul = 1$]
	Let $C$ be a positive number such that $K$ is contained in the $\ell^2$-ball centred at $0$ with radius $C$.
	 Let $x \in K$. Assume there exists $j \geq 2$ with $|x_j| > \eps $ in the Jordan block. 
	 Let $N = \frac{1}{k} \left (\frac{k^2C}{\eps}\right )^{2^{k-1}}$. Then there exists $j$ and $t \leq N$ with $|x_j(t)| > C$.
\end{lemma}
\begin{proof}
	
	Given the set of equations 
	\[x_j (t) = \sum_{i=j}^{k} \binom{t}{i-j}\Lambda^{t-(i-j)}x_{i},\] 
	we want an $N$ such that there exists $j$ such that $|x_j(N)| > C$.
	
	Let $j_1 \geq 2$ be the smallest $j$ such that $|x_{j_1}| > \eps$ (we are given that such a component exists).
	Consider the component 
	\[x_{j_1 -1} (t) = \sum_{i={j_1 -1}}^{k}  \binom{t}{i-j_1 +1}\Lambda^{t-(i-j_1 +1)}x_{i}.\]
	We set $N_{j_1} = kC/\eps$.
	Observe that (note $\modul=1$) 
	\[|x_{j_1 -1} (N_{j_1})| \geq | (kC/\eps) x_{j_1}  | - |x_{j_1 -1}| - \sum_{i={j_1 +1}}^{k} \left |\binom{N_{j_1}}{(i-j_1 +1)}\Lambda^{t-(i-j_1 +1)}x_{i}\right |\]
	Since the first term is larger than $kC$ and the second term is smaller than $C$, the only way $|x_{j_1 -1} (N_{j_1})|$ could be less than $C$ (and thus not escape $K$) is if one of the later terms is larger than $C$. 
	Let $j_2$ be the highest index such that $\left |x_{j_2}\binom{N_{j_1}}{{j_2}-j_1 +1}\right |\geq C.$ Note that $j_2 > j_1$.
	We now have a lower bound on a higher index coefficient, namely
	$\left |x_{j_2} \right |\binom{N_{j_1}}{{j_2}-j_1 +1}\geq C$.
	We now repeat the process with the component 
	\[x_{j_2 -1} (t) = \Lambda^tx_{j_2 -1}+ \binom{t}{1} \Lambda^{t-1}x_{j_2} + \sum_{i={j_2 +1}}^{k} \binom{t}{i-j_2 +1}\Lambda^{t-(i-j_2 +1)}x_{i}\]
	We have $\left |x_{j_2} \right |\binom{N_{j_1}}{{j_2}-j_1 +1}\geq C$ thus setting $N_{j_2} > k\binom{N_{j_1}}{j_2-j_1 +1}$ ensures that $|N_{j_2}x_{j_2}| > kC$.
	
	%	\red{What if instead of going on another component, we stayed on the same one and just waited more? For example waiting a time $2kT_{j_1}$? The first two component would be ignored and we could get some result similar to before no? And it seems more efficient. J: I remember trying this. You have more terms to deal with, and I think the bound was worse.}
	
	Continuing this process, we will either find a component that escapes the set
	or move on to a component with higher index, which can happen at most $k-1$ times, because we have the constraints $j_1 \geq 2, \forall m, j_m \leq k, $ and $j_m > j_{m-1}$.
	This gives us a recursive definition for the bound, which is \[N_{j_m} > k\binom{N_{j_{m-1}}}{j_m-j_{m-1} +1} \]
	
	We wish to find an upper bound on $N = N_M$, the time by which we are guaranteed that 
	at least one component escapes, subject to the constraints $j_1 \geq 2, j_M \leq k, $ and $j_m > j_{m-1}$.
	We can solve the recursive inequality by weakening it (since we only need an upper bound on $N$) to \[N_{j_m} > (kN_{j_{m-1}})^{j_m-j_{m-1} +1}.\]
	Note that pulling the constant $k$ into the exponentiated part is valid because $j_m-j_{m-1} +1	>2$ always.
	Setting $S_{j_m} = kN_{j_m}$, we get
	$S_{j_m} > S_{j_{m-1}}^{j_m-j_{m-1} +1}, S_{j_1}= k^2C/\eps,$	
	which reduces to \[S_{j_M} > \left (\frac{k^2C}{\eps}\right )^{\prod_{m=2}^{M} (j_m -j_{m-1} + 1)}\]
	The term $\prod_{m=2}^{M} (j_m -j_{m-1} + 1)$ is maximised when for all $m$, $j_m = j_{m-1} +1$, thus in the worst case we have 
	\[S_{j_M} > \left (\frac{k^2C}{\eps}\right )^{2^{k-1}}\]
	%	Using the AM-GM inequality, we note that $$\prod_{n=2}^{N} (j_n -j_{n-1} + 1) < \left ( \frac{\sum_{n=2}^{N} (j_n -j_{n-1} + 1)}{N} \right )^N = \left ( \frac{N-1 + j_N -j_1}{N} \right )^N $$
	%	
	%	Since $N \leq k-1$, we observe that $$\left ( \frac{N-1 + j_N -j_1)}{N} \right )^N \leq \left ( \frac{2(k-1)}{k-1} \right )^{k-1} = 2^{k-1}$$
	Thus we have a bound for a modulus-1-eigenvalue component to escape, which is 
	\[\boxed{N = \frac{1}{k} \left (\frac{k^2C}{\eps}\right )^{2^{k-1}} .}\]
\end{proof}
	
\begin{lemma}[Exponentially expanding Jordan block - $\modul > 1$]
	Let $C$ be a positive number such that $K$ is contained in the $\ell^2$-ball centred at $0$ with radius $C$.
	Let $x \in K$. Assume there exists $j \geq 2$ with $|x_j| > \eps $ in the Jordan block. 
	Let $N =2^{k-1}\frac{\log(kC/\eps)}{\log \modul}$. Then there exists $j$ and $t \leq N$ with $|x_j(t)| > C$.
\end{lemma}
\begin{proof}
	The proof is very similar in structure to the modulus-1-eigenvalue case, though the presence of an exponential factor gives us a much better bound.

	By construction of $\eps$, there exists $j_1 \geq 2$ such that $|x_{j_1}| > \eps$.
	Consider the component \[x_{j_1} (t) = \sum_{i={j_1}}^{k} \binom{t}{i-j_1}\Lambda^{t-(i-j_1)}x_{i}.\]
	Set $N_{j_1} = \frac{\log(kC/\eps)}{\log \modul}$ and observe that \[|x_{j_1} (N_{j_1})| \geq |\Lambda^{\frac{\log(kC/\eps)}{\log \modul}}x_{j_1}| - \sum_{i={j_1 +1}}^{k}  \left|\binom{N_{j_1}}{i-j_1}\Lambda^{N_{j_1} - (i-j_1)}x_{i}\right|.\]
	Since the first term is larger than $kC$, the only way $|x_{j_1} (N_{j_1})|$ can be less than $C$ (and thus not escape the set) is if one of the later terms is larger than $C$. Let $j_2$ be the highest index such that 
	$\left|\binom{N_{j_1}}{j_2-j_1}\Lambda^{N_{j_1} - (j_2-j_1)}x_{j_2}\right|\geq C.$ Note that $j_2 > j_1$.
	We now have a lower bound on a higher index coefficient, namely $\binom{N_{j_1}}{j_2-j_1}\modul^{N_{j_1} - (j_2-j_1)}\left|x_{j_2}\right|\geq C.$
	Now we repeat the process with the component \[x_{j_2} (t) =\Lambda^t x_{j_2}+ \sum_{i={j_2 +1}}^{k} \binom{t}{i-j_2}\Lambda^{t-(i-j_2)}x_{i}\]
	We want $|\Lambda^t x_{j_2}| > kC$, so it is enough to set 
	\[\modul^{N_{j_2}} > k \binom{N_{j_1}}{j_2-j_1} \modul^{N_{j_1}-(j_2 -j_1)}  .\]
	Similarly to the previous lemma, this gives us a recursive definition for the bound, which is 
	\[N_{j_m} - N_{j_{m - 1}} >  \frac{\log\binom{N_{j_{m-1}}}{j_m - j_{m-1}}}{\log \modul} + \frac{\log k}{\log \modul} - (j_2 - j_1).\]
	
	For $N$ sufficiently larger than $k$ (which we may easily assume), we can weaken this inequality to 
	\[N_{j_m}  > 2N_{j_{m - 1}} + \frac{\log k}{\log \modul} ,\]
	which is easily solved to get 
	\[N_M > 2^{M-1}N_{j_1} -\frac{\log k}{\log \modul}.\]
	
	As $M \leq k$ and $N_{J_1} = \frac{\log(kC/\eps)}{\log \modul}$
	we have a bound for a positive-eigenvalue component to escape, which is 
	\[\boxed{N \leq 2^{k-1}\frac{\log(kC/\eps)}{\log \modul} .}\]
\end{proof}

\begin{lemma}[Exponentially shrinking Jordan block: $\modul < 1$]
	Let $C \geq n$ be a positive number such that $K$ is contained in the $\ell^2$-ball centred at $0$ with radius $C$.
	Let $x \in K$. Let $\eps$ be a positive real number. Let $N =\frac{4k}{\log(1/\modul)
	}\log\left(\frac{2k}{\log(1/\modul)}\right) + \frac{2\log 
	(kC/\varepsilon)}{\log(1/\modul)}$. Then there exists $t \leq N$ with $|x_j(t)| < \eps$ for all $j$ in the block.
\end{lemma}
\begin{proof}
	If $\modul = 0$ then we have $x_j(n) = 0$ for all $j$ in the block, so that we may assume $\gamma > 0$ for the rest of the proof.

	We have the equations \[x_j (t) = \sum_{i=j}^{k} \binom{t}{i-j}\Lambda^{t-(i-j)}x_{i},\]
	For any $j$, we have the result 
	\[|x_j (t)| \leq \sum_{i=j}^{k} \binom{t}{i-j}\cdot|\Lambda^{t-(i-j)}x_{i}| < kC (et/k)^k \modul^{t-k},\] where the second inequality is obtained via standard bounds on binomial coefficients
	
	In order to have $|x_j (t)| < \varepsilon$, it is enough to have $kC (et/k)^k \modul^{t-k} < \varepsilon$, which is equivalent (after weakening slightly by dropping irrelevant terms )to
	$t > \frac{k}{\log (1/\modul)} \log t + \frac{\log (kC/\varepsilon) }{\log (1/\modul)}$.
	
	Here we need a small technical lemma.
	
	\begin{lemma}
		[{\cite[Lemma A.1, Lemma A.2]{shalev2014understanding}}]
		\label{lem:techi}
		Suppose $a \geq 1$ and $b > 0,$ then $t \geq a \log t + b$ if $t \geq 4a\log(2a) + 2b$.
	\end{lemma}
	
	Applying this lemma we get a bound on $N$ such that for all $j\leq k$, $x_{j}(N)<\varepsilon$, namely 
	\[\boxed{N \leq \frac{4k}{\log(1/\modul)
		}\log\left(\frac{2k}{\log(1/\modul)}\right) + \frac{2\log 
		 (kC/\varepsilon)}{\log(1/\modul)}.}\]
\end{proof}

We now show that the time to leave $K_{\geq \eps}$ is doubly exponential in the ambient dimension, singly exponential in the rest of the input data, and inverse polynomial in $\eps$.
\begin{lemma}[Non-recurrent overall bound]
	\label{lem:realcase}
	There exists a positive integer constant $L$ with the following property:
    
    Let $K$ be a semialgebraic set of complexity at most $(n,d,\tau)$.
    Let $A \in \mathbb{A}^{n \times n}$ be a matrix in real Jordan normal form.
	Assume that the characteristic polynomial of $A$ has rational coefficients, bounded in bitsize by $\tau$.
	Let $\eps$ be a positive real number. 
	Define a partition of $K$ into $K_{\operatorname{rec}}, K_{< \eps }$ and $K_{\geq \eps}$ as described in Section \ref{Section: Preliminaries}.
	Let 
	\[ 
		N_{\geq \eps} = \left(\frac{1}{\eps}\right)^{2^n} \cdot 2^{(\tau \cdot d)^{Ln}}.
	\]  
	Then for all $x \in K_{\geq \eps}$ there exists $t \leq N_{\geq \eps} $ such that $A^tx \notin K_{\geq \eps}$, which is to say, we have either escaped the set $K$ completely or moved into $K_{< \eps} \cup K_{\operatorname{rec}}$.
\end{lemma}

\begin{proof}
Since $x \in K_{\geq \eps}$, we start with at least one component greater than $\eps$ in Jordan norm. This allows us to leverage the block-wise bounds.

	Let $N_{\geq \eps} = 2 \max_\modul \{ N_\modul \}$, where $N_\modul$ ranges over the possible bounds depending on the size of the eigenvalues of $A$.
	Within a time $N_{\geq \eps} /2 = \max_\modul \{ N_\modul \}$, thanks to the analysis
	of the block-wise bounds, there are three possibilities:
	\begin{itemize}
		\item the orbit has increased in size beyond $C$ and has thus left $K$. This occurs if there was a component associated to an expanding eigenvalue that was larger than $\eps$ in Jordan norm;
		\item all components are now smaller than $\eps$, thus leaving $K_{\geq \eps}$ (by entering $K_{< \eps}$ );
		\item some component corresponding to an expanding eigenvalue which was originally 
		less than $\eps$ has become greater than $\eps$. In this case, waiting another 
		$N_{\geq \eps}/2$ amount of time puts the trajectory in the first case, ensuring it escapes.	
	\end{itemize}
	Thus in all cases the trajectory has escaped by time $N_{\geq \eps}$.

	We now explicitly compute $N_{\geq \eps}$ by using the complexity of $K$ and $A$ to bound the three possibilities, namely 
	\[\boxed{N = \frac{4k}{\log(1/\modul)
		}\log\left(\frac{2k}{\log(1/\modul)}\right) + \frac{2\log 
		(kC/\varepsilon)}{\log(1/\modul)}.}\](shrinking eigenvalue $\modul$),
\[\boxed{N = 2^{k-1}\frac{\log(kC/\eps)}{\log \modul} .}\] (exponentially expanding eigenvalue $\modul$) and 
\[\boxed{N = \frac{1}{k} \left (\frac{k^2C}{\eps}\right )^{2^{k-1}} .}\] (modulus 1 eigenvalue), where $k$ is the multiplicity of the Jordan block and thus $k \leq n$, the dimension.

We now compute bounds on $\log \modul$ and $C$ in terms of $K$ and $A$.

\proofsubparagraph{Bounding $\log \modul$:} 

Let $\tau$ be a bound on the bitsize of the coefficients of the minimal polynomial of the 
eigenvalues of $A$ as well as a bound on the total bitsize of $K$.

Using the fact that $x/2 < \log(1 + x) < x$ for $|x| << 1$, Lemma \ref{Lemma: bounding modulus of eigenvalues away from 1} yields a constant $R$ satisfying $\frac{1}{\log \modul} < 
2^{\left(\tau n\right)^{R}}$.
%(H \cdot n)^{n^R}$. \EN{TODO: potentially express this in terms of $\tau$ and bring closer to main text}

\proofsubparagraph{Bounding $C$:} 

Let $d$ bound the degree of the polynomials defining $K$. Then from Theorem \ref{Theorem: Basu-Roy radius bound} we have the existence of a constant $S$ satisfying $C < 2^{(\tau \cdot d)^{S(n+1)}}$. 

Plugging these bounds into the iteration bounds from the previous lemmas, and overapproximating for simplicity, we finally get the following bound:
With a new constant $Q$ based on $R$ and $S$, we have
\[N_{\geq \eps} \leq \left(\frac{1}{\eps}\right)^{2^n} \cdot 2^{(2^n \cdot(\tau \cdot d)^{S(n+1)}) }+ (2 \cdot \tau \cdot d)^{(\tau n)^Q} + \log(1/\eps) \cdot 2^{(\tau n)^Q}.\]

We can simplify this still further by amalgamating terms.
Letting $L$ be a new constant, we set
\[\boxed{N_{\geq \eps} = \left(\frac{1}{\eps}\right)^{2^n} \cdot 2^{(\tau \cdot d)^{Ln}} }.\]
Thus the time to leave $K_{\geq \eps}$ is doubly exponential in the dimension, singly exponential in the rest of the input data, and inverse polynomial in $\eps$.
\end{proof}

%% file: app-example.tex
\section{Example of matching lower bound}
\label{Appendix: example}

In Section~\ref{Section: Lower Bounds}, we matched the bound using a rotational 
system which needed a doubly exponential time to escape by the small hole in the circle.
Here, we present another example where the doubly exponential bound comes from the size of 
the set we define.

\begin{example}
The construction of our first family of instances 
$(K_{(n,d,\tau)},A_{(n,d,\tau)})_{(n,d,\tau)\in\N^3}$ 
relies on the fact that one can define a compact semialgebraic set whose size is doubly-exponential in the ambient dimension.

For $(n,d,\tau) \in \N^3$, define $K_{(n,d,\tau)} \subseteq \R^{n + 1}$ as the set of all points
$(x_1,\dots,x_n,x_u)$ satisfying the (in)equalities:
\begin{align*}
& x_u = 1,\\
& x_1 = 2^{\tau},\\
& \text{For }1\leq i \leq n - 2,\; x_{i+1} = x_i^d, \\
& 0\leq x_{n} \leq x_{n - 1}^d.
\end{align*}
Thus, a point $x \in \R^{n + 1}$ belongs to $K_{(n,d,\tau)}$ if and only if it is of the form 
$\left(2^{\tau}, 2^{\tau d},\dots, 2^{2^{\tau d^{n - 2}}}, y, 1\right)$
where $y \in \left[0, 2^{\tau d^{n - 1}}\right]$.

We now define $A_{(n,d,\tau)}$ to be the matrix which only adds 1 (through the coefficient $x_u$) to the penultimate coordinate:
\[
A_{(n,d,\tau)} = 
\begin{pmatrix}
					1 & 0 &  \dots      &  & \\
					0 & 1  &  \dots      &  & \\ 
					\vdots	& \vdots     &\ddots  &  & \\ 
						&      &        & 1 & 1 \\
						&      &        & 0 & 1 \\
\end{pmatrix}
\]
Therefore, given an initial point $x_0\in K_j$, we have that $x_t=A_{(n,d,\tau)}^tx_0= x_0 + (0,\dots,0,t,0)$. 
This sequence obviously escapes. 
The point $(2^\tau,2^{\tau d},\dots,2^{\tau d^{n - 2}}, 0, 1) \in K_{(n,d,\tau)}$ requires $2^{\tau d^{n - 1} + 1}$ iterations to escape.
This is doubly exponential in the ambient dimension and singly exponential in the rest of the data.
\end{example}